\documentclass[10pt]{article}
\usepackage{header}

\ifpdf
\hypersetup{
  pdftitle={The Deep Parametric PDE Method:
  Application to Option Pricing},
  pdfauthor={K. Glau and L. Wunderlich}
}
\fi

\title{The Deep Parametric PDE Method:\\
  Application to Option Pricing\thanks{This work was funded 	by the Alan Turing Institute and  EPSRC grant no.  EP/T004738/1.}}

\author{Kathrin Glau\thanks{School of Mathematical Sciences, Queen Mary University of London, Mile End Road, London E1 4NS, United Kingdom 
  ( k.glau@qmul.ac.uk, \l.wunderlich@qmul.ac.uk).} \and Linus Wunderlich\footnotemark[2]}
  
\usepackage{graphicx}
\usepackage{amsthm}
\newtheorem{theorem}{Theorem}
\begin{document}

\maketitle

\begin{abstract}
We propose the deep parametric PDE method to solve  high-dimensional parametric partial differential equations. A single neural network approximates the solution of a whole family of PDEs after being trained without the need of sample solutions. 
As a practical application, we  compute option prices in the multivariate Black-Scholes model.  After a single training phase, the   prices for different time, state and model parameters  are available in milliseconds.  We evaluate the accuracy in the price and a generalisation of the implied volatility 
with examples of up to 25 dimensions.
A comparison with alternative machine learning approaches, confirms the effectiveness of the approach. 

\end{abstract}


\section{Introduction}
Solving parametric partial differential equations (PDE) in high dimensions is a   major challenge in many areas of science and engineering. In particular, it is of high importance for repetitive tasks in finance, such as 
real-time risk monitoring, uncertainty quantification and credit value adjustments. Frequent requests of derivative prices, require an efficient evaluation  at different times, states and model parameters. 
Modern methods in option pricing based on Monte-Carlo are of limited efficiency. Other methods, such as based on PDEs or transforms, suffer from the curse of dimensionality.

Deep neural networks (DNN) offer efficient approximations with no curse of dimensionality~\cite{grohs:jentzen:18, hutzenthaler:jentzen:20}. 
Thus, we apply them to numerically solve high-dimensional PDEs. 
While evaluating a neural network is fast, the training phase is computationally complex. 
To fully exploit the benefit of DNNs in the PDE context, we propose the \emph{deep parametric PDE method} to learn the solution for all parameters simultaneously.

For an unsupervised training of the network, we first need a formulation of the parametric PDE  as an appropriate optimisation problem. 
The least-squares formulation of a PDE yields a suitable loss function,  as also used in the deep Galerkin method~\cite{sirignano:17} for a fixed parameter set. 
The formulation naturally extends to the parametric setting and we propose a universal approach that  allows for wide applications. As example we compute multi-asset option prices in the Black-Scholes model.\footnote{We provide an implementation of the deep parametric PDE method for two assets at \url{https://github.com/LWunderlich/DeepPDE/blob/main/TwoAssetsExample/DeepParametricPDEExample.ipynb}
	}
	 Further practical applications  are pricing in stochastic volatility or jump-diffusion models, where Monte-Carlo methods are still state of the art for high dimensions, and examples beyond finance.

The deep parametric PDE method exhibits a natural offline-online decomposition. In the one-time offline phase, the neural network is trained. Then in the online phase, evaluating the solution for any time, state and parameter value is simply the matter of evaluating a single neural network.

\subsection{Literature Review}
There is a large research effort in medium- and high-dimensional option pricing. Classical methods include different variants of Monte-Carlo methods~\cite{giles:15, lecuyer:09}, Fourier pricing~\cite{glau:10},  sparse grid  integration~\cite{bayer:18, griebel:10, holtz:11} and low-rank approximations~\cite{glau:20} (see also references therein). Also PDE-based solvers are considered, e.g. using an operator splitting~\cite{hout:16}, via expansion~\cite{reisinger:18}, wavelets~\cite{hilber:13} and radial basis function~\cite{pettersson:08}. Monte-Carlo methods typically lack efficiency, while the other methods suffer from the curse of dimensionality in high-dimensional settings.

A possible remedy is the application of deep neural networks, which  have drastically improved the field of artificial intelligence in the recent years~\cite{lecun:15}. 
The idea to use neural networks in finance was already researched in the 1990, e.g., in~\cite{hutchinson:94, malliaris:93} based on supervised learning. 
 Also the use of shallow neural networks as a discretisation of PDEs using the Galerkin method was experimented with, e.g., in~\cite{barucci:97, meade:94}. 
A recent literature review~\cite{ruf:20} provides an overview over applications of neural networks in option pricing and hedging.

In general, the approaches can be classified as either unsupervised or supervised. 
An overview of unsupervised  PDE-based approaches  can be found in~\cite{vidales:19}.
The deep BSDE method uses a reformulation of the problem as a backward stochastic differential equation (BSDE), which is then solved by a neural network, e.g,.~\cite{beck:jentzen:19, chan:19, han:jentzen:18, hure:20}. 
A method related to the proposed deep parametric PDE method is the deep Galerkin method introduced in~\cite{sirignano:17}. Among the applications considered are the partial integro-differential equations and Hamilton-Jacobi-Bellman equations~\cite{al-aradi:19, al-aradi:18} as well as general Stokes equations~\cite{li:20:dgm}.
Physics-informed neural networks use a similar approach~\cite{chen:hesthaven:20, guofei:19, raissi:19}.
To the best of our knowledge, no approach to solve general parametric PDEs directly using a neural network has been published yet. 

An alternative approach to directly solving a PDE with neural networks is  applying supervised learning to sample data, e.g.,~\cite{liu:osterlee:19}. 
The deep Kolmogorov method recently introduced in~\cite{berner:20} applies a hybrid approach using supervised learning and SDE-based techniques. There, samples of the underlying stochastic process of a Kolmogorov PDE are used to train a neural network.

There have been developed supervised learning approaches to parametric PDE problems, unrelated to finance. 
Examples are to approximate a low-dimensional quantity of interest~\cite{khoo:20} or the whole solution mapping~\cite{bhattacharya:20, li:20:fourier_network}. 


\subsection{Main Contribution}
We summarise the main contributions of this article.
\begin{itemize}
\item We introduce the deep parametric PDE method to solve a family of PDEs with a single network training. After a one-time offline phase to train the network, the solution for different parameter values can be evaluated in milliseconds.
\item With an application to multi-asset option pricing, we show the practical use of the method. 
\item We provide a general proof of convergence which includes non-smooth data, as given for the option pricing problem.
\item To evaluate the performance of the deep parametric PDE method, we study the error in the option price and in the implied volatility for problems of up to $25$ dimensions.
\end{itemize}

The article is structured as follows. In Section~\ref{sec:dgm}, we introduce the deep parametric PDE method for parabolic problems. We specify the formulation for option pricing in the multivariate Black-Scholes model. Also, we show convergence of the deep parametric PDE method under general assumptions. Further details for the option pricing problem, i.e., the parameter dependency and reference solvers, are given in Section~\ref{sec:option_pricing_details}. 
Details regarding the implementation, including feature-scaling and hyper-parameter tuning, are described in Section~\ref{sec:implementational_details}. We investigate numerical examples with four to 25 dimensional problems in Section~\ref{sec:numerics}. Finally Section~\ref{sec:conclusion} summarises and concludes the article.
\section{The Deep Parametric PDE Method} \label{sec:dgm}

We consider a parametric parabolic PDE on the domain $(0,T)\times \Omega$, where $\Omega \subset \R^d$ with $d\in \mathbb{N}_{>0}$ is bounded with a smooth boundary. For each parameter $\mu$ in the compact parameter domain $\Param$, we solve for $u(\cdot;\mu)$, such that
\begin{subequations} \label{eq:pde_introduction}
\begin{align}
\partial_\timet u(\timet,x;\mu) + \DiffOpx u(\timet,x;\mu) &= f(\timet, x; \mu), \quad& (\timet, x)& \in \Q = (0,T)\times \Omega,\\
u(0,x;\mu) &= g(x; \mu), \quad& x&\in \Omega, \\
u(\timet, x; \mu) &= \usigma(t,x;\mu), \quad& (\timet, x)&\in\Sigma = (0,T) \times \partial\Omega.
\end{align}
\end{subequations}	
Here $\DiffOpx$ is a strongly elliptic differential operator of second order operating on the state variable $x$. It is parametrised by $\mu\in\Param$ and we assume a continuous parameter dependency.
Also the parameter dependency on $f(\mu)\in L^2(\Q)$, $g(\mu)\in L^2(\Omega)$ and $\usigma(\mu)\in L^2(0,T, H^{1/2}(\Sigma))$ is assumed to be continuous. Note that we frequently abbreviate $f(\cdot;\mu)$ as $f(\mu)$. 
Solvability and uniqueness of the solution for each parameter is given by, e.g.,~\cite[Chapter 4, Equation 15.38]{lions_v2:72}.

We consider standard Lebesgue and Sobolev spaces, as introduced in~\cite{lions_v1:72}.
$L^2$ denotes the space of square integrable functions, in the case of $L^2(\Q)$ mapping from  $\Q$ to $\R$ and for $L^2(0,T,X)$  mapping from $(0,T)$ to the Hilbert space $X$.
$H^1(\Omega)$ is the space of square-integrable functions on $\Omega$, whose first weak derivatives are also square-integrable.  Its trace space is $H^{1/2}(\partial \Omega)$.

\subsection{Overview of the Deep Parametric PDE Method}
We propose the deep parametric PDE method as an  approximation of the solution $u$ by a single neural network. After training the network, the approximate solution is given  for all times, states and parameter values. 
In order to truly exploit the given structure, we  determine the loss function $\Loss(u)$ purely by the PDE. A suitable approach is based on a least-squares formulation of the PDE. To ease the learning process, we transform the solution to an auxiliary one, which is of a similar magnitude throughout the domain. This auxiliary solution is then approximated using a deep neural network with the time, state and parameter variables as the input. 

The complete procedure is summarized as follows:
\begin{itemize}
\item In a one-time offline phase, which is the computationally expensive part of the method, we train the neural network by minimising the PDE's residuals.
\item
In the online phase, the solution to the PDE problem for any time, state and parameter value is obtained by evaluating the neural network, which is computationally fast.
\end{itemize}
Key advantage of the approach is that with a single training, an approximate solution is available for all considered PDEs simultaneously. This is in contrast to state-of-the-art deep learning approaches to solve PDEs, where one training phase yields the solution of a single PDE. Including the parameters of the problem into the neural network thus enables us to fully exploit the potential of deep learning to approximate high dimensional functions. The investment in a  computationally expensive training phase thus pays off, since it yields a fast solution for the complete family of PDEs.

\subsection{Choice of the Loss Function}
To  approximate~\eqref{eq:pde_introduction} by a neural network, we first need to define an appropriate minimisation problem. In the deep parametric PDE method, we use a least-squares formulation of the PDE. To account for the parameter-dependency, another integral over the parameter domain $\Param$ is added.

For a given function $u\colon \Q\times\Param\rightarrow \R$ of sufficient smoothness, we define the loss based on the PDE's residuals:
\begin{align}\label{eq:residual}
\Loss(u) = \LossInt(u) + \LossIC(u) + \weightBC \LossBC(u) .
\end{align}
The interior residual is defined as
\begin{align*}
\LossInt(u) =  
|\Q\times\Param|^{-1}
\int_{\Param} \int_{\Q} \left( \partial_\timet u(\timet,x;\mu) + \DiffOpx u(\timet,x;\mu) -f(\timet, x; \mu) \right)^2 ~\mathrm{d}(\timet,x) ~\mathrm{d}\mu,
\end{align*}
with $|\Q\times\Param|$ the size of the domain, and the initial residual as
\begin{align*}
\LossIC(u) = 
|\Omega\times\Param|^{-1}
\int_{\Param} \int_\Omega \left( u(0, x; \mu) - g(x; \mu) \right)^2 ~\mathrm{d}x ~\mathrm{d}\mu.
\end{align*}
The boundary residual 
\begin{align*}
\LossBC(u) = 
|\Sigma\times\Param|^{-1}
\int_{\Param} \| u( \mu)  - u_\Sigma(\mu)  \|_{H^{1/2}(\Sigma)} ^2 \mathrm{d}\mu
\end{align*}
is weighted by a factor $\weightBC\geq 0$, which we will choose as zero in our experiments for simplicity.
Unlike  approaches based on supervised learning of expensively computed samples (e.g.,~\cite{liu:osterlee:19}), no samples are required as this approach is \emph{unsupervised}.

The integrals are numerically evaluated by Monte-Carlo quadrature, which yields a similarity to mean squared error-residuals often used in machine learning:
\begin{align} \label{eq:discrete_residual}
\LossInt(u) &\approx 
\sum_{i=1}^N \Big( \partial_\timet u\big(\timet^{(i)}\! ,x^{(i)};\mu^{(i)}\big) + \DiffOpx u\big(\timet^{(i)}\!,x^{(i)};\mu^{(i)}\big) -
f\big(\timet^{(i)}\!,x^{(i)};\mu^{(i)}\big) \Big)^2 /N, \\
\LossIC(u) &\approx 
\sum_{i=1}^N \Big( u\big(0, \hat x^{(i)}; \hat \mu^{(i)}\big) - g\big(\hat x^{(i)}; \hat \mu^{(i)}\big) \Big)^2 /N,\notag
\end{align}
where $(\timet^{(i)}\!, x^{(i)}\!, \mu^{(i)}) \in \Q\times\Param$ and $(\hat x^{(i)}\!, \hat \mu^{(i)})\in \Omega \times \Param$ for $i=1,\ldots, N$ are chosen randomly with a uniform distribution.   In our experiment, we choose $N=10,000$ and observed less accurate approximations with smaller values of $N$. 
Approximating the boundary residual $\LossBC$ would be more involved. A practical approach could be to replace the $H^{1/2}(\Sigma)$ norm, by the $L^2(\Sigma)$-norm and to perform the same Monte-Carlo quadrature. As our experiments show good results without the term, we omit it for simplicity.

\subsection{Multivariate Option Pricing in the Black-Scholes model} 
We apply the deep parametric PDE method to an option pricing problem in the Black-Scholes model.
Expressing the option price in logarithmic asset variables, $u(\timet, x; \mu)$ denotes the fair price of an option at time to maturity $\timet$ for the asset prices $s_i=e^{x_i}$:
\begin{align}
u(\timet,  x; \mu) &= \operatorname{Price}(T-\timet, e^x; \mu),\notag \\
\operatorname{Price}(\physicalt, s; \mu) &= e^{- r (T-\physicalt) } \ev(G(S_T(\mu))\,| \,S_{\physicalt}(\mu)=s),
\label{eq:original_problem_setting}
\end{align}
with $d$ underlyings $S_\physicalt(\mu) = (S_\physicalt^1(\mu), \ldots, S_\physicalt^d(\mu))$ and the physical time $\physicalt = T - \timet$. $G$ denotes the payoff function at maturity and the assets $S$ are modelled by a multivariate geometric Brownian motion.

The parameter $\mu$ can describe model parameters as well as option parameters.
In our setting the parameter vector $\mu$ contains the risk-free rate of return, volatilities and correlations, each with a smooth parameter dependency. 

In the Black-Scholes model, the differential equation~\eqref{eq:pde_introduction} is homogeneous, i.e., $f(t,x;\mu)=0$ and the  operator reads
 \begin{align*}
\DiffOpx u(\timet,x;\mu) = &r  u(\timet, x; \mu) - \sum_{i=1}^d \left(r-\frac{\sigma_i^2}{2}\right)\partial_{x_i} u(\timet, x;\mu) - \sum_{i,j=1}^d  \frac{\rho_{ij} \sigma_i \sigma_j}{2}  \, \partial_{x_ix_j} u(\timet,x;\mu),
\end{align*}
with $r = r(\mu)$, $\sigma_i = \sigma_i(\mu)$ and $\rho_{ij} = \rho_{ij}(\mu)$ with $\rho_{ii} = 1$.
We choose the domain as a hypercube: $\Omega = (\xmin, \xmax)^d$. We note that the boundary of the domain exhibits less regularity than assumed in our original setting, but we do not expect any issues arising from this.

For our main experiments, we consider European basket call options with equal weights and fixed strike price $K$:
\[
g(x) = G(e^x) = \left(\frac{1}{d}\sum_{i=1}^d e^{x_i} - K\right)_+ 
= \max\left\{0,\, 
\frac{1}{d} \sum_{i=1}^d e^{x_i} - K \right\}.
\]
As typical for option pricing, the initial condition is not smooth. However, in this case it is still in $H^1(\Omega)$. Thanks to the smoothing property of parabolic PDEs, this is sufficient regularity for the approximation results shown in Section~\ref{sec:approximation}. 

For different parts of the boundary $\partial\Omega$, the boundary values $u_\Sigma$ could be in principle chosen as the average asset price or the solution of a lower-dimensional option pricing problem. Since our numerical results are well even without the term, we omit the extra computational effort which would be required.

 \subsection{Localisation of the PDE}
We can   decompose the option price in two parts: the \emph{trivial no-arbitrage bound},
 and the remaining time value of the option.
 The no-arbitrage bound equals the maximum of zero and the expected payoff of an auxiliary derivative. This  derivative shares all specifications of the basket option, despite that the owner is forced to execute at maturity.
In the case of a single asset, this bound captures the asymptotic behaviour of the option. As a consequence,  the  remaining time value   is bounded and thus well suited for its approximation by a neural network. 
 
However, the no-arbitrage bound in general is not smooth, making it unsuitable for a transformation of the PDE. Instead, we  consider a smooth approximation $\localisation \in C^\infty$. In case of put and call options, the non-smoothness stems from the maximum function, which can be approximated excellently by the softplus function.	 For European basked call option, we thus have
\begin{align}\label{eq:localisation}
\localisation(\timet,x; \mu) = \frac{1}{\lambda}\log\left(1+e^{\lambda\left( \frac{1}{d}\sum_{i=1}^d e^{x_i}/d -  K e^{-r\timet} \right)}\right),\quad \text{ for } \lambda > 0.
\end{align}
Note that for $\lambda \rightarrow \infty$, we approach the no-arbitrage bound:
\[
\lim_{\lambda\rightarrow\infty} \localisation(\timet, x; \mu) = \left( \sum_{i=1}^d e^{x_i}/d -  K e^{-r\timet} \right)_+.
\]
 The drawback when using large values of $\lambda$ is that the second derivative  can become too large. Therefore in practice, we use a medium value, $\lambda=0.1$.

In the deep parametric PDE method, we are left to approximate the \emph{residual   value}
$u(t, x;\mu) - \localisation(t, x;\mu)$
 by a  neural network.

\subsection{Neural Network}
We use a variant of highway networks~\cite{highway_networks:15} that proved successful in the approximation of PDEs in~\cite{sirignano:17}. After an initial dense layer, several gated layers are applied and finally combined to a scalar output. Denoting the input variables $(\timet, x; \mu)$ as $h^0 \in \R^{n}$ with $n=1+d+ n_\mu$, we have the first dense layer as
\[
h^1 = \psi( W^0h^0 + b^0) \in \R^m,
\]
with $W^0\in \R^{m\times n}, b^0\in \R^m$ for $m$ nodes in each layer. The activation function $\psi$ is the element-wise application of a smooth function, in our case the hyperbolic tangent. 

Then for $l=1,\ldots,L$ with $L$ the number of layers, we have
\[
h^{l+1} = (1 - g^l) \odot h^{l+1/2} + z^l \odot h^l\in \R^m,
\]
with gates $g^l$ and $z^l$ which can pass $h^l$ and the intermediate layer computation $h^{l+1/2}$.
The operator $\odot$ denotes element-wise multiplication of vectors. 
 The intermediate layer computation $h^{l+1/2}$ includes an additional gate $r^l$ which can drop the information present in the previous layer,
\[
h^{l+1/2} = \psi(U^{h,l} h^0 + W^{h,l} (h^l \odot r^l ) + b^{H,l} )\in\R^m.
\]
Each of the three gates have the same standard structure:
\begin{align*}
g^l &= \psi(U^{g,l} h^0 + W^{g,l} h^l + b^{g,l} ),\\
r^l &= \psi(U^{r,l} h^0 + W^{r,l} h^l + b^{r,l} ),\\
z^l &= \psi(U^{z,l} h^0 + W^{z,l} h^l + b^{z,l} ).
\end{align*}
In all cases the weights and biases are of the same size $U^{*,l} \in \R^{m\times n}$, $W^{*,l} \in \R^{m\times m}$ and $b^{*,l}\in \R^m$ and are trainable parameters. 

A final dense layer and adding the localisation yields the trial functions for the option price:
\[
\uDNN^\theta(\timet,x;\mu) =  
 \localisation(\timet,x;\mu) + W^{L+1} h^{L+1} + b^{L+1},
\]
with $W^{L+1} \in \R^{1\times m}$, $b^{L+1} \in \R$ and the vector $\theta$ collecting all  trainable \emph{network parameters} $W^*, U^*, b^*$.
We seek the  network parameters  $\hat \theta$ which minimise the loss:
\begin{align}\label{eq:minimisation_of_loss}
\hat \theta = \operatorname*{arg\,min}_{\theta} \Loss(\uDNN^\theta).
\end{align}
This defines the deep parametric PDE solution:
\[
u(\timet,x;\mu) \approx \uDGM(\timet,x;\mu) = \uDNN^{\widehat\theta}(\timet,x;\mu).
\]
The resulting function $\uDGM$ is a  single neural network that approximates the true solution of the PDE at all times, states and parameter values simultaneously. With one training, we have solved the whole family of PDEs.

\subsection{Approximation Properties}
\label{sec:approximation}
In~\cite{sirignano:17}, convergence of the deep Galerkin method is shown for fixed parameters and with  smooth solutions. In option pricing, we often deal with non-smooth initial conditions, so we adapt the proof to this case.
We assume that the optimisation problem~\eqref{eq:minimisation_of_loss} is solved correctly. Research on convergence of numerical optimisers for neural networks is an emerging field, having shown first results, see e.g.,~\cite{bercher:jentzen:20}.
We expect future results to allow us to push our analysis further in this direction.

We show convergence in two steps: First, we establish the existence of  a neural network that minimises the loss function up to any prescribed accuracy. Then, we prove that any approximation with a loss function smaller than $\varepsilon > 0$ has an $L^2$-error less than $c \varepsilon$.

We denote the set of single layer-neural networks as
\[
\boldsymbol{\mathcal{C}} = \bigcup_{m=1}^\infty \{ \zeta: \zeta(\timet, x; \mu) = \sum_{j=1}^m \beta_j \psi(w^\timet_j \timet + {w^x_j}^\top x + {w^\mu_j}^\top \mu + b_j)\} \subset C^\infty(\Q\times\Param).
\]
Our notation of function spaces follows~\cite{lions_v2:72, lions_v1:72}, where
$
H^{r, s}(\Q) = L^2(0,T, H^r(\Q))\cap H^s(0,T,L^2(\Q))
$
and 
$H^r$ is the  Sobolev space of order $r\in\R$, both being defined in~\cite[Chapter 1]{lions_v1:72}.
We also consider by $C$ continuous functions and by $C^\infty$ smooth functions. 

We consider the case with active boundary conditions, i.e., $\weightBC=1$. The boundary conditions need to be sufficiently regular. An example would be a sequence of lower-dimensional Black-Scholes solutions. Regularity in this case can be shown by iterating this process up to the case of a single asset.
Also the initial and boundary conditions are required to continuously intersect, which results in the compatibility condition
$\left. g(\mu) \right|_{ \partial\Omega} =
\left. \usigma(\mu) \right|_{\{0\}\times \partial\Omega} 
$.

\begin{theorem} \label{thm:existence_of_approximation}
Let $g\in C(\Param, H^1(\Omega))$ and $\usigma\in  C(\Param, H^{3/2,1}(\Sigma))$ fulfil the compatibility condition $\left. g(\mu) \right|_{ \partial\Omega} =
\left. \usigma(\mu) \right|_{\{0\}\times \partial\Omega}$ for each parameter $\mu\in\Param$ and let $f\in C(\Param, L^2(\Q))$. Also let the parameter-dependency of the differential operator $\DiffOpx$ be continuous in the operator norm from $H^2(\Omega)$ to $L^2(\Omega)$.

Then, for any $\varepsilon > 0$, there exists a function $\uDNN \in \boldsymbol{\mathcal{C}}$,  such that $\Loss(\uDNN) < \varepsilon$.
\end{theorem}
\begin{proof}
We use $H^{2,1}(\Q)$-regularity of parabolic PDEs as shown in~\cite[Chapter 4, Equation 15.39]{lions_v2:72}. While this result holds for each $\mu\in\Param$, the analyticity of the inversion of a  linear operator (see e.g.,~\cite[Corollary A.2]{eldering:13}) shows integrability on the compact domain $\Param$, i.e., $u\in L^2(\Param, H^{2,1}(\Q))$.
Although this is sufficient regularity to show the desired approximation property, we could not find a direct results on the approximability of functions in anisotropic Sobolev spaces by neural networks.  Therefore, we first approximate the solution by a smooth function, which we then approximate by a neural network. 

Using a density argument, we can show that there exists $\usmooth \in C^{\infty}(\Q\times\Param)$, such that 
\[
\| u - \usmooth \|_{L^2(\Param, H^{2,1}(\Q))}^2 \leq \varepsilon.
\]
Now we can use~\cite[Theorem 3]{hornik:91}, which shows that there exists a neural network $\uDNN \in \boldsymbol{\mathcal{C}}$ with
\[
\| \usmooth - \uDNN \|_{H^2(\Q\times\Param)}^2 \leq \varepsilon.
\]
This implies a close approximation also of $u$: 
\begin{align*}
\| u - \uDNN \|_{L^2(\Param, H^{2,1}(\Q))}^2  &\leq 
c\| u - \usmooth \|_{L^2(\Param, H^{2,1}(\Q))}^2 +
c\| \uDNN - \usmooth \|_{L^2(\Param, H^{2,1}(\Q))}^2 \\&\leq 
c\varepsilon.
\end{align*}
Then the interior residual can be rewritten and estimated as
\begin{align*}
\|\partial_t \uDNN + \DiffOpx \uDNN - f \|_{L^2(\Q\times\Param)}^2 =
\|\partial_t (\uDNN - u) + \DiffOpx (\uDNN - u) \|_{L^2(\Q\times\Param) }^2\leq c \varepsilon.
\end{align*}
Trace estimates  imply
\begin{align*}
\| \usigma - \uDNN \|_{L^2(\Param, H^{1/2}(\Sigma))}^2 &=
\| u - \uDNN \|_{L^2(\Param, H^{1/2}(\Sigma))}^2 \leq c \varepsilon, \\
\| g - \uDNN \|_{L^2(\{0\}\times\Omega\times\Param)}^2&=
\| u - \uDNN \|_{L^2(\{0\}\times\Omega\times\Param)}^2 \leq c \varepsilon,
\end{align*}
which concludes $\Loss(\uDNN) \leq c \varepsilon$.
\end{proof}

In the second step we show that for any smooth function with a small loss, the $L^2$-error is bounded. 
The proof is fairly standard and relies on the stability of the PDE, which holds uniformly in~$\Param$.

\begin{theorem}\label{thm:stability}
Let the assumptions of Theorem~\ref{thm:existence_of_approximation} hold. Additionally let the parameter-dependency of the differential operator $\DiffOpx$ be continuous in the operator norm from $H^1(\Omega)$ to $H^{-1}(\Omega)$.

Then, there exists a constant $c<\infty$, such that for all $\uDNN\in C^{\infty}(\Q\times\Param)$ it holds
\[
\|u - \uDNN\|_{L^2(\Q\times\Param)}^2 \leq c \, \Loss(\uDNN).
\]
\end{theorem}
\begin{proof}
First, we rewrite the loss function as an average over the parameter space
\begin{align*}
\Loss(\uDNN) &= 
\int_\Param h(\mu) ~\mathrm{d}\mu, \text{ where}
\\
h(\mu) &=\| \uDNN(\mu)  + \DiffOpx \uDNN(\mu) -f(\mu) \|_{L^2(\Q)}^2 +
\| \uDNN(\mu) - u(\mu) \|_{L^2(\{0\}\times\Omega )}^2 \\&\qquad+
\| \uDNN(\mu) - u(\mu) \|_{ H^{1/2}(\Sigma)}^2.
\end{align*}
Note that as $u(\mu)\in H^{2,1}(\Q)$, we have sufficient regularity to replace $g(\mu)$ and $\usigma(\mu)$ by $u(\mu)$ in these integrals.

We consider the residual equation for the error $e(\mu) = u(\mu)-\uDNN(\mu)$:
\begin{align*}
\partial_\timet e(\mu) + \DiffOpx e(\mu) &= r(\mu), \quad \text{ on } \Q,\\
e(\mu) &= u(\mu) - \uDNN(\mu), \quad \text{ on } \{0\}\times \Omega,\\
e(\mu) &= u(\mu) - \uDNN(\mu), \quad \text{ on } \Sigma,
\end{align*}
where $r(\mu) = f(\mu) - \partial_\timet \uDNN(\mu) - \DiffOpx \uDNN(\mu)$.
Stability of parabolic PDEs as shown in~\cite[Chapter 4, Equation 15.38]{lions_v2:72} yields
 \begin{align*}
\| e(\mu) \|_{L^2(\Q)}^2& \leq
c\| r(\mu) \|_{L^2(\Q)}^2 + c\| u(\mu) - \uDNN(\mu) \|_{H^{1/2}(\Gamma)}^2 \\&\qquad+ c\|u(\mu) - \uDNN(\mu)\|_{L^2(\{0\}\times\Omega)}^2 
= c h(\mu),
\end{align*}
where we have used $\| e(\mu) \|_{L^2(\Q)}^2 \leq
\| e(\mu) \|_{H^{1,0}(\Q)}^2$ and 
$\| r(\mu) \|_{H^{-1,0}(\Q)}^2 \leq 
\| r(\mu) \|_{L^2(\Q)}^2$.
Again, as the inversion of a linear operator is analytic, the constant is bounded over $\Param$ and the right hand side is thus integrable, which yields:
\begin{align*}
\| u - \uDNN \|_{L^2(\Q\times\Param)}^2 = \int_\Param \| u(\mu) - \uDNN(\mu) \|_{L^2(\Q)}^2 \,\mathrm{d}\mu \leq c\! \int_\Param h(\mu) \,\mathrm{d}\mu = c\, \Loss(\uDNN).
\end{align*}
\end{proof}
\section{Details Regarding Option Pricing} \label{sec:option_pricing_details}
In this section, we provide the precise parametrisation of the option pricing problem and the evaluation of reference pricers. 

\subsection{Parametrisation of the Option Pricing Problem}
In this section, we discuss the dependence of the Black-Scholes model on the vector of parameters $\mu$. The potential parameters of   basket  options in the Black Scholes model are the strike price $K$, the risk-free rate of return $r$, the volatilities $\sigma_i$ and the correlations $\rho_{ij}$. 
We note that by rescaling the asset prices, we can easily transform the problem into an option pricing problem with a fixed strike price $K$. For this reason, we do not consider the strike price as a parameter, but fix it at $100$. 
The risk-free rate of return $r$ and the volatilities $\sigma_i$ are each an entry of the parameter vector. 

For $d>2$ parametrising the correlation matrix $(\rho_{ij})_{i,j}$ is a non-trivial task as the resulting covariance matrix has to be symmetric and positive semi-definite. A naive approach would be to parametrise the covariance matrix based on its Cholesky factors, leading $d(d-1)/2$ parameters, which are difficult to interpret.  
Instead, we consider a parametrised model based on the   practical approach   suggested in~\cite{doust:10}, where a valid correlation matrix is computed from pairwise correlations. The inputs are $(d-1)$ independent pairwise correlations $\hat\rho_i = \rho_{i,i+1}$ and the missing entries are calculated by successive products: $\rho_{ij} = \rho_{ji} = \prod_{k=i}^{j-1} \hat\rho_k$, for $j>i$.
Positive definiteness for $\hat\rho_i\in(-1,1)$ is shown in~\cite{doust:10} by providing a Cholesky decomposition. 
Even for noisy or polluted  estimates, this approach yields a valid correlation matrix. 

In summary, the parameter vector is given as 
\[
\mu = (r, \sigma_1, \ldots, \sigma_d, \hat\rho_1, \ldots, \hat\rho_{d-1}),
\]
which yields a smooth parameter-dependency of the pricing problem~\eqref{eq:pde_introduction} in the sense of Theorems~\ref{thm:existence_of_approximation} and~\ref{thm:stability}.
We note that for $\sigma_i > 0$ and $\hat\rho_i \in (-1,1)$, the differential operator $\DiffOpx$ is strongly elliptic in $x$, thus we assume
\[
\Param \subset \R\times (0, \infty)^d \times (-1,1)^{d-1}.
\]
With these $n_\mu = 2d$ parameters, the overall dimension for a European basket call option with $d$ underlyings is $n=3d+1$. 

\subsection{Error Evaluation with the Implied Volatility for Basket Options}\label{sec:implied_volatility}
For single-asset options, the implied volatility is the standard measure to compare prices and accuracies over a range of different products. 
There is no unique extension of the concept to a multivariate case. Here, we propose a natural formulation of the implied volatility for basket options. Each arbitrage-free option price is mapped to a unique volatility and computationally it boils down to the univariate case. 

Given the   price $c$ (either the exact price $u(\timet, x; \mu)$ or the approximated one $\uDGM(\timet, x; \mu)$), we define the implied volatility as $\sigmaiv >0$, such that
\[
\BS\left(\timet, \sum_{i=1}^d e^{x_i}/d, r, \sigmaiv, K\right) = c.
\]
A value for the implied volatility can be computed for any value in between the two trivial no-arbitrage bounds $c > c_{\rm lb} = \left( \sum_{i=1}^d e^{x_i}/d - K e^{-r\timet}\right)_+$  and 
$c < c_{\rm ub} = \sum_{i=1}^d e^{x_i}/d$.
The implied volatility is a good measure for the relative accuracy, as its  values are of a similar magnitude over all sizes of the asset prices. The proposed extension inherits this property.

\subsection{Reference Pricers}
\label{sec:reference_pricer}
To evaluate the performance of the deep parametric PDE method, we need to compare it to alternative pricers. These reference pricers may be more expensive to evaluate for many parameters, as they only serve for comparison. 

While for vanilla European basket options, the Black-Scholes formula 
\begin{align}\label{eq:bs}
c(\timet, x; \mu) = \BS(\timet, e^x, r, \sigma, K) &= \Phi(d_1) e^x - \Phi(d_2) K e^{-r \timet},
\end{align}
with $d_1 = \frac{1}{\sigma \sqrt{\timet}} \left(x-\ln(K) + r\timet + \frac{\sigma^2 \timet}{2}\right)$ and $d_2 = d_1 - \sigma \sqrt{\timet}$,
provides an explicit option price, this is no longer available for basket options.
We present a reference pricer for basket options as well as an academic example of a basket option with an explicit solution. 

\subsubsection{Smoothing the Payoff of European Basket Call Options}
We can evaluate the option price by integrating a smoothened payoff, as developed in~\cite{bayer:18} and extended in~\cite{poetz:20}. After a variable transformation, the $d$-dimensional problem of computing the option price is split in a one-dimensional and a smooth $(d-1)$-dimensional  problem. The first part can be solved precisely and the second part is solved using Gau\ss{}-Hermite quadrature. 
For convenience of the reader, we recapitulate the key steps.

Decomposing the covariance matrix as outlined in~\cite[Lemma 3.1]{bayer:18} yields $\lambda_i$ and $(v_{i,j})_{ij}$, such that for independent $Y_i\in \N(0,\lambda_i^2)$ the stochastic process of the 
logarithmic prices is
$
\log( S_T^i(\mu) ) = x_i + ( r-\sigma_i^2/2)\timet + Y_1 + \sum_{j=2}^d v_{i,j} Y_j,
$
with $x_i$ the logarithmic asset price at time-to-maturity $\timet$. 
Solving a conditional expectation for $Y_1$ given $Y_2, \ldots, Y_d$ first, yields the option price as a $(d-1)$-dimensional problem with a smooth payoff function:
\[
c(\timet, x; \mu) = \ev \left( \BS(1, h(Y_2,\ldots, Y_d), 0, \lambda_1, e^{-r\timet}K)\right),
\]
where 
$
h(Y_2, \ldots, Y_d) = 
\frac{1}{d}\sum_{i=1}^d e^{x_i-
\sigma^2\timet/2 } e^{\sum_{j=2}^d v_{i,j} Y_j}
$.
As the  function $h$ is smooth and the dimension is reduced by one, we have a simpler problem to solve than~\eqref{eq:original_problem_setting}. Particularly a Gau\ss{}-Hermite quadrature as proposed in~\cite{poetz:20} suits well and is used here as a reference pricer. 
Because still a high-dimensional integral needs to be computed for each call, we only use it to validate our approach and do not propose it as an alternative parametric solver.

\subsubsection{Special Cases With an Explicit Solution} \label{sec:geometric_payoff}
Inspired by~\cite{sirignano:17}, we consider basket options with a payoff based on the geometric mean, which allows for an analytical solution of the parametric option price. Instead of the algebraic mean, the geometric mean $(\prod_{i=1}^d e^{x_i})^{1/d} $  enters the payoff:
\begin{align}\label{eq:geometric_payoff}
g(x)  = \left( e^{\sum_{i=1}^dx_i/d}  - K\right)_+.
\end{align}
The geometric mean of a multivariate geometric Brownian motion 
is a univariate stochastic process of the same type.
Thus the high-dimensional problem is reduced to a single-asset pricing problem
with a dividend. 
For further simplicity, we consider underlyings of equal correlations and volatilities in which case the total dimension of the parametric problem is $n=d+4$. 
The solution in this special case is given as 
\begin{gather}\label{eq:geometric_solution}
u(x, \timet; \mu)  = N(d_1) e^{-q\timet} e^{\bar x} - N(d_2) Ke^{-r\timet}, 
\end{gather}
with $  d_1  = \frac{1}{\bar \sigma\sqrt{\timet}}\left( \bar x  -\ln(K) + r \timet -q \timet + \frac{\bar \sigma^2}{2}\timet\right) $, 
$d_2 = d_1 - \bar \sigma\sqrt{\timet}$ and 
 $\bar x = \sum_{i=1}^d x_i/d$, 
$\bar\sigma^2 = \sigma^2/d (1 + (d-1) \rho)$ and $q =  \sigma^2/2 - \bar \sigma^2/2$. 
We   also adapt the extension of the implied volatility for basket options as described
in Section~\ref{sec:implied_volatility} to this case.

\section{Implementational Details}\label{sec:implementational_details}
We implement the proposed deep parametric PDE   method
	 in \texttt{python} using the machine-learning package \texttt{keras}~\cite{keras}, based on \texttt{tensorflow}~\cite{tensorflow} (the versions used are \texttt{python 3.6.3}, \texttt{keras 2.4.0} and \texttt{tensorflow 2.3.0}).
The models are trained on the GPU nodes of the high performance computing cluster at Queen Mary, Apocrita~\cite{apocrita}. 
The used GPUs are    Nvidia Tesla K80 and V100. They are accessed using \texttt{CUDA 10.1.243}.
The trained neural networks are evaluated on an end-user device, a 2015 MacBook Pro with a $2.7$GHz dual-core CPU and $8$GB RAM.
In the following subsections, we provide details regarding the algorithms to optimise trainable parameters as well as hyper-parameters.

\subsection{Minimising the Loss}
The minimisation of the loss function with respect to the weights of the neural network is a highly non-linear non-convex optimisation problem. As such we solve it iteratively by a variant of a gradient-descent method. 
We use an early stopping criterion to determine when the numerical optimiser has converged. Where the loss has not improved after 50 epochs (with 10 batches each), optimisation is stopped and the weights with the minimal observed loss are used.
For a good generalisation of the solution, we re-sample the points used in the discrete loss~\eqref{eq:discrete_residual} after each batch. This results in a good approximation on average of the loss function by the discrete loss.

A common problem when optimising neural networks are vanishing gradients, where a node saturates and the derivative with respect to its weights becomes numerically zero. To avoid this from happening, the initial value of the weights is of a high importance~\cite{glorot:10}.
Most standard initialisations of the weights in neural networks assume input to be  normalised. An input of the order of $100$, as for the asset prices, would lead to vanishing gradients and thus extremely slow training. Thus, we linearly transform the computational domain to $[-1, 1]$ for time, asset values and each parameter. 
All weights are initialised using the Glorot normal initializer~\cite{glorot:10}.

\subsection{Hyper-parameter Optimisation}
In addition to the weights and biases collected in the vector $\theta$ that will be optimised, neural networks have hyper-parameters, such as the number of layers and the number of nodes per layer. When these are chosen arbitrarily, results cannot be expected to be optimal. 
We use \texttt{keras-tuner}~\cite{kerastuner} for a hyper-parameter optimisation based on a random search. We set the limits for the neural network to be between $2$ and $10$ layers with  $10$ to $110$ nodes each. Different optimisers are tested, and the learning rate (i.e., the step-size of the optimiser) is set to be between $0.01$ and $0.0001$.

While the precise values differed in different situations, we found that $9$ layers with $90$ nodes each provide good results in all of our cases. As the optimiser we choose Adam~\cite{AdamOpt} with a learning rate of $0.001$.

\section{Numerical Examples}\label{sec:numerics}
In this section, we  consider different option pricing problems with one to eight underlyings. We look at the error in the price and the implied volatility in different situations, but also look into the convergence of the Adam optimiser.

If not stated otherwise, we consider volatility values   between $10\%$ and $30\%$, pairwise correlations between $0.2$ and $0.8$ and risk-free rates of return from $10\%$ up to $30\%$.  The maximal time to maturity considered is $4$ in all cases, with the minimal time of interest being $0.5$. 
The strike price is fixed to $100$, which is no limitation thanks to re-scaling properties. The values of interest for the underlying assets range from $25$ to $150$. 		
To limit the truncation error, we consider a larger computational domain with asset prices between $21$ and $460$. For the parameters, the computational domain and the domain of interest   coincide. 
Unless stated otherwise, we use the default parameter values $r=20\%$, $\sigma_i=20\%$ and $\hat\rho_{i}=0.5$. 

In all examples, we use the same network architecture with $9$ layers and $90$ nodes per layer. This has the key advantage for applications, that no extra hyper-parameter optimisation is required to apply the method to a new situation. The numerical results confirm the required robustness.

We start with the univariate case in order to validate the method against the explicit solution. Then, we consider basket options in different scenarios and compare it to the reference pricer introduced in Section~\ref{sec:reference_pricer} and alternative machine learning methods.  
To further explore the versatility of our proposed method, we also consider the geometric payoff and briefly discuss the computation of Greeks. 
\subsection{Single-Assets Call Options}
As an initial study, we consider the case of European call options with one underlying as we can compare the solution to the Black-Scholes formula~\eqref{eq:bs}.

\begin{figure}[htbp]
\begin{minipage}{0.29\textwidth}
\begin{center}
\includegraphics[width=\textwidth]{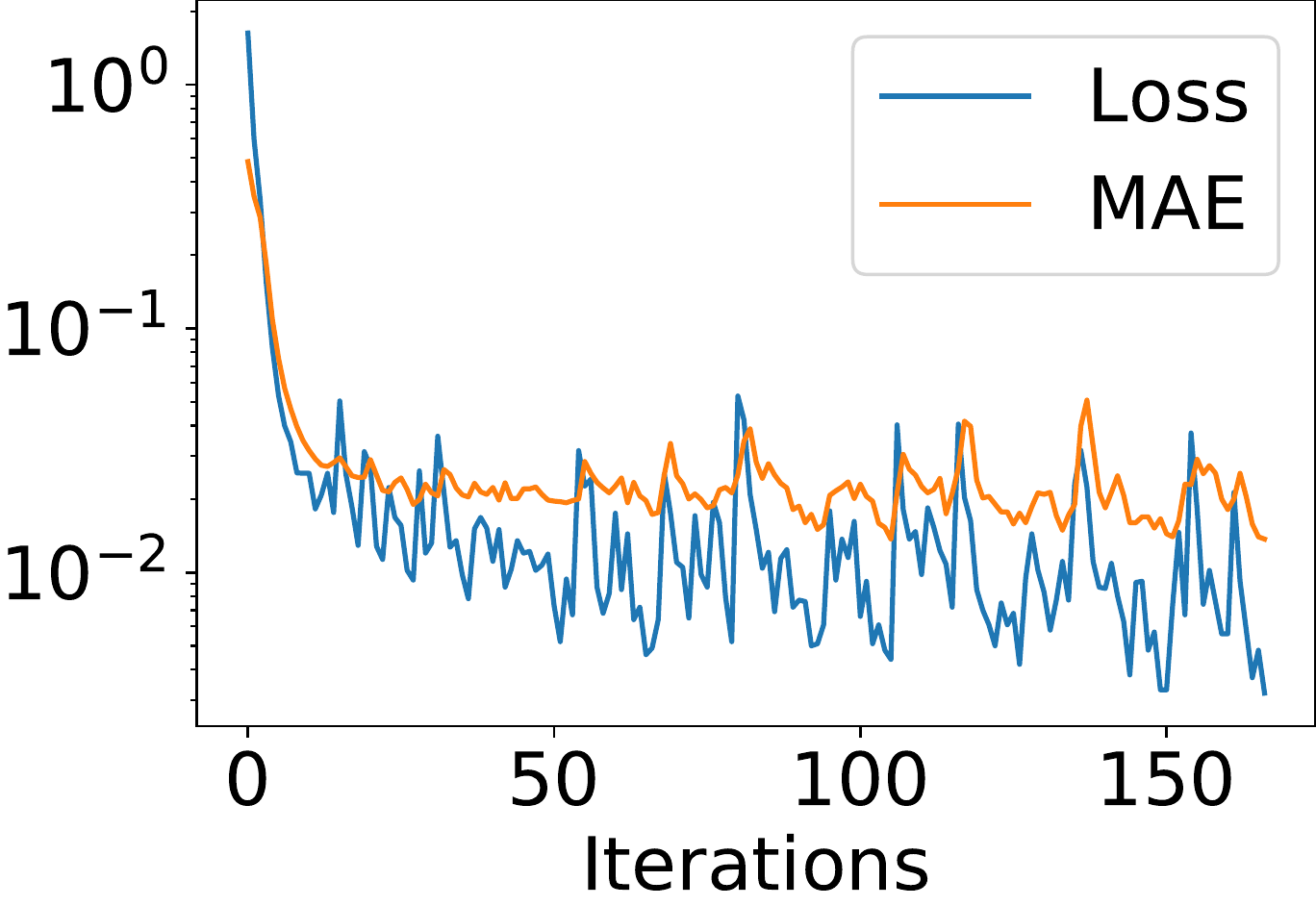}
\caption{Convergence of loss and mean absolute error (MAE) over iteration steps of the optimisation. Only the convergence up to the optimal iteration is shown. $50$ more epochs were computed but resulted in no improvement of the loss and are not shown. }
\label{fig:convergence_1d}
\end{center}
\end{minipage}
\hspace{1em}
\begin{minipage}{0.67\textwidth}
\begin{center}
\includegraphics[width=.48\textwidth]{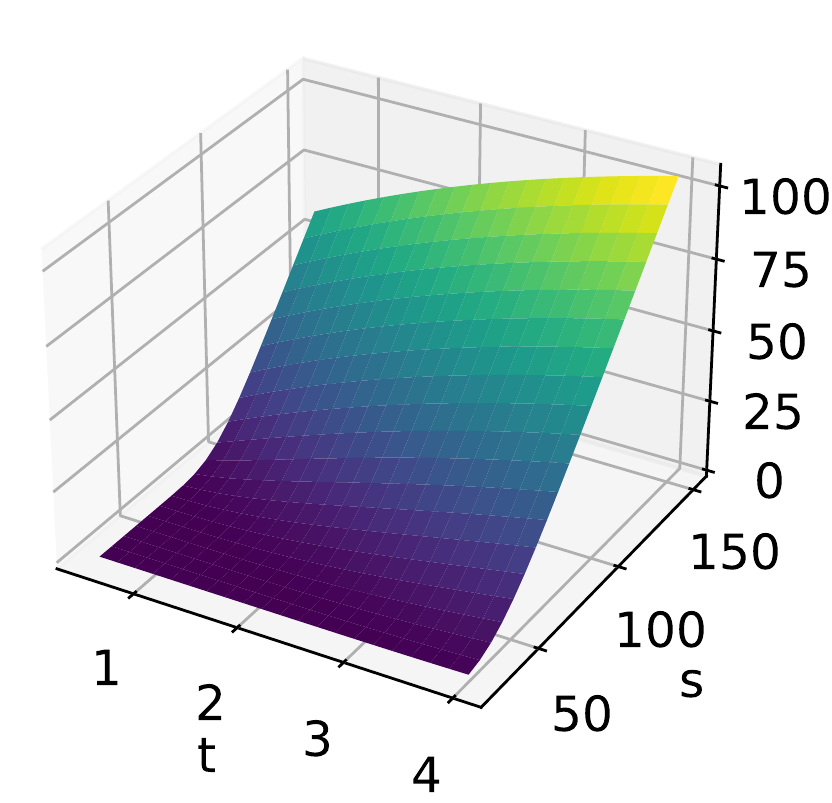} \hspace{0.1em}
\includegraphics[width=.48\textwidth]{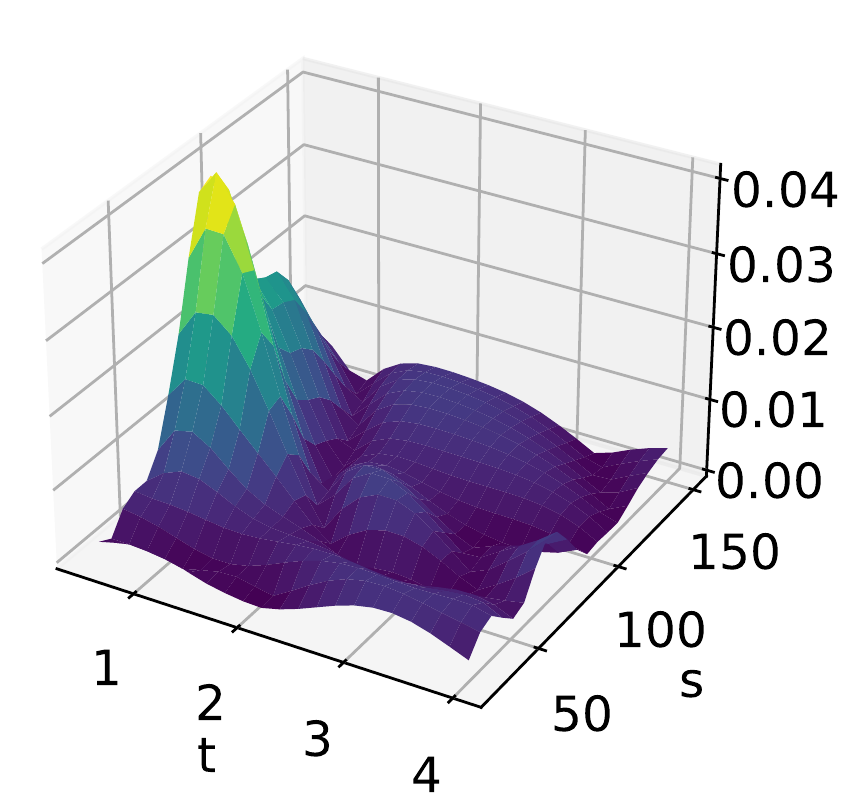}\\
\includegraphics[width=.48\textwidth]{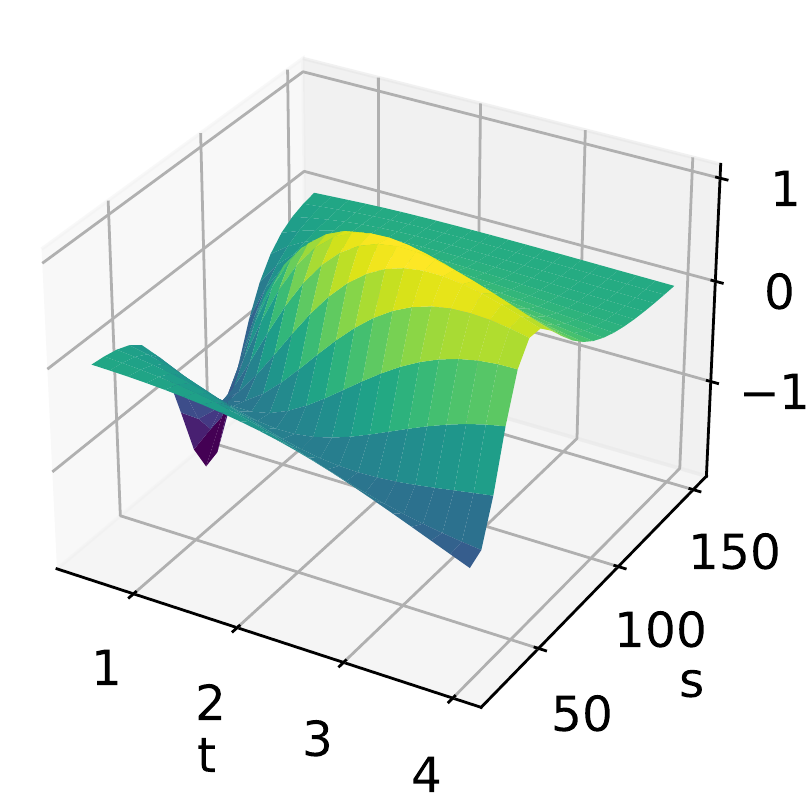}
\hspace{0.1em}
\includegraphics[width=.48\textwidth]{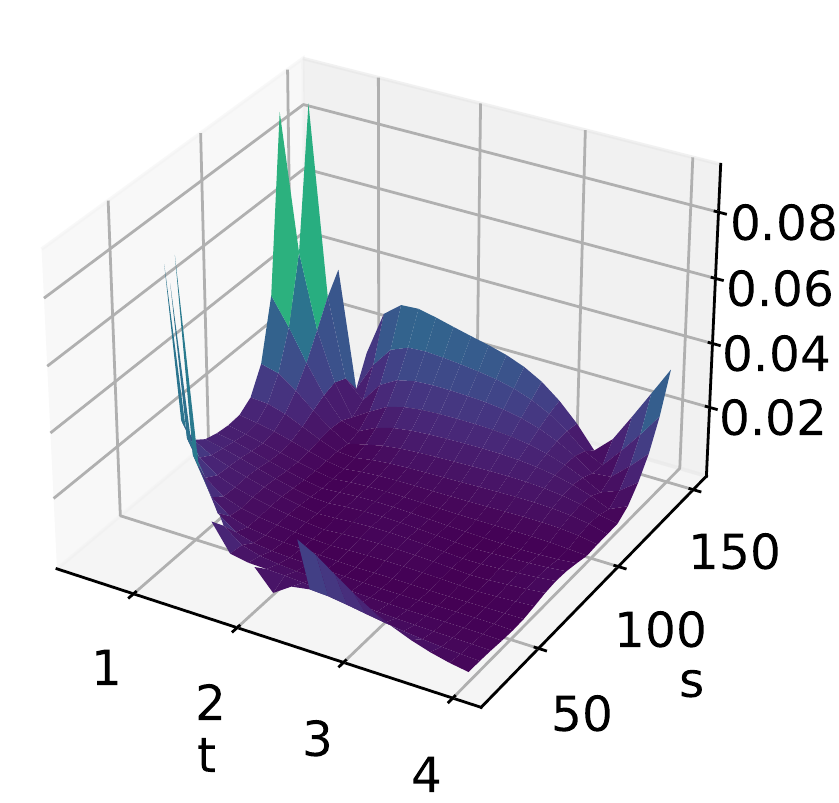}
\caption{Top row: One-dimensional deep parametric PDE solution (left) and errors (right). Bottom row: Residual   value of the solution (left) and  relative error of the implied volatility (right). 
All values shown for fixed parameters.}
\label{fig:solution_numerical_and_error_1d}
\end{center}
\end{minipage}
\end{figure}

Convergence of the residual and the validation error computed on a fixed set of $10,000$ random points are shown in Figure~\ref{fig:convergence_1d}. We see a clear improvement of the approximation over the iterations in both the loss function and the error, which shows the suitability of the considered loss functions and confirms the results of Theorem~\ref{thm:stability}.
In particular, we can see a clear relationship between the value of the loss function and the error on the validation set. This confirms that it is reasonable to use the loss function to detect convergence of the optimiser, as the validation error will not always be available in practice. 

Figure~\ref{fig:solution_numerical_and_error_1d} shows the approximation and the error over the domain of interest for  fixed parameters.
We see an overall well approximation of the option price with absolute errors considerably lower than $0.05$. While the option price is between $0$ and $100$ in the domain of interest, the localisation captures most parts successfully. The residual   value, which is approximated by the neural network, is betweens $-1.5$ to $1$. 

Also depicted in Figure~\ref{fig:solution_numerical_and_error_1d} is the relative error of the implied volatility. 
As the implied volatility can be very sensitive to the option price, we
 only compute it where the difference between the exact option price and the lower trivial no-arbitrage bound $c_{\rm{lb}}$ (defined in Section~\ref{sec:implied_volatility}) is larger than $0.005$. In most parts of the domain, the implied volatility is accurate with a relative error less than $1\%$, overall less than $10\%$. 
We see that the relative error of the implied volatility peaks for small time to maturities and far in or out of the money, in which cases the implied volatility is too sensitive to be a reasonable error measure.  
After the successful results with one underlying, we next consider European basket options.

\subsection{Basket Call Options}
We consider basket options with between two and eight assets. First, we study the error at fixed parameter values and then focus on the parameter dependency. Finally, we compare the proposed deep parametric PDE method to alternative machine learning approaches. 

\subsubsection{Evaluation for Fixed Parameters}\label{subsubsec:fixed_param}
We first consider small baskets with two and three assets and then larger ones with five and eight assets. As the visualisation of high dimensional functions is challenging, we use different approaches in different dimensions. 
While in the examples of this section the parameters are fixed, we stress that the same solution can be used to evaluate any parameter in the domain of interest. 

With two underlyings, we inspect the solution at the maximal time to maturity for fixed parameters, varying the price of the assets. 
Figure~\ref{fig:solution_numerical_and_error_2d} shows the approximation and the error over the domain of interest.
We see  a good approximation with error values below $0.1$ in the whole domain. Again, we note that the residual   value of the option is of a small magnitude with values varying between $-2$ and  $2$. 

\begin{figure}[htbp]
\begin{center}
\includegraphics[width=.3\textwidth]{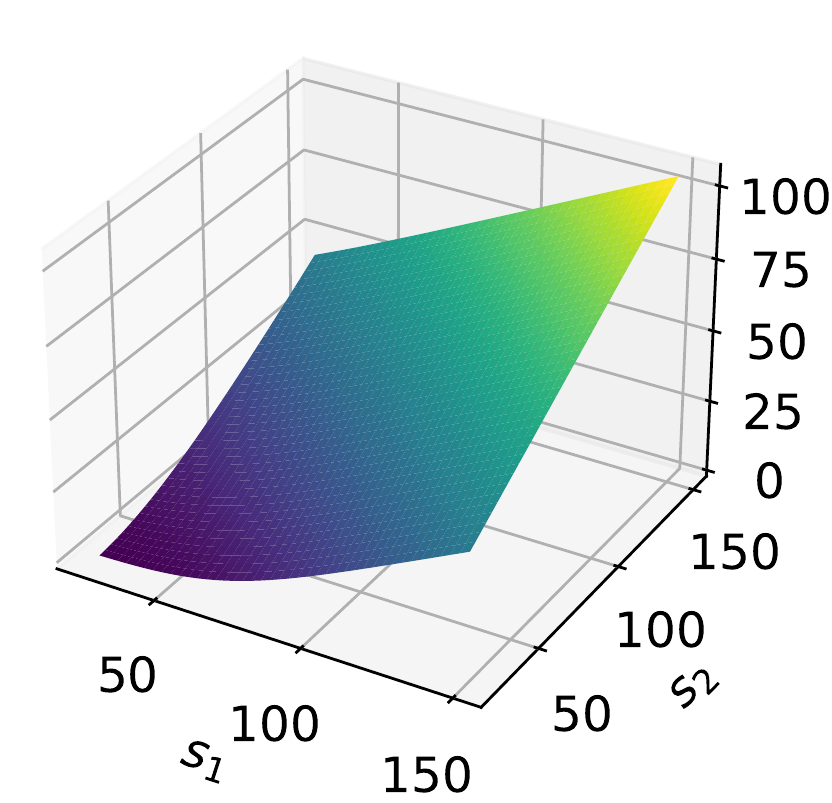}\hspace{1em}
\includegraphics[width=.3\textwidth]{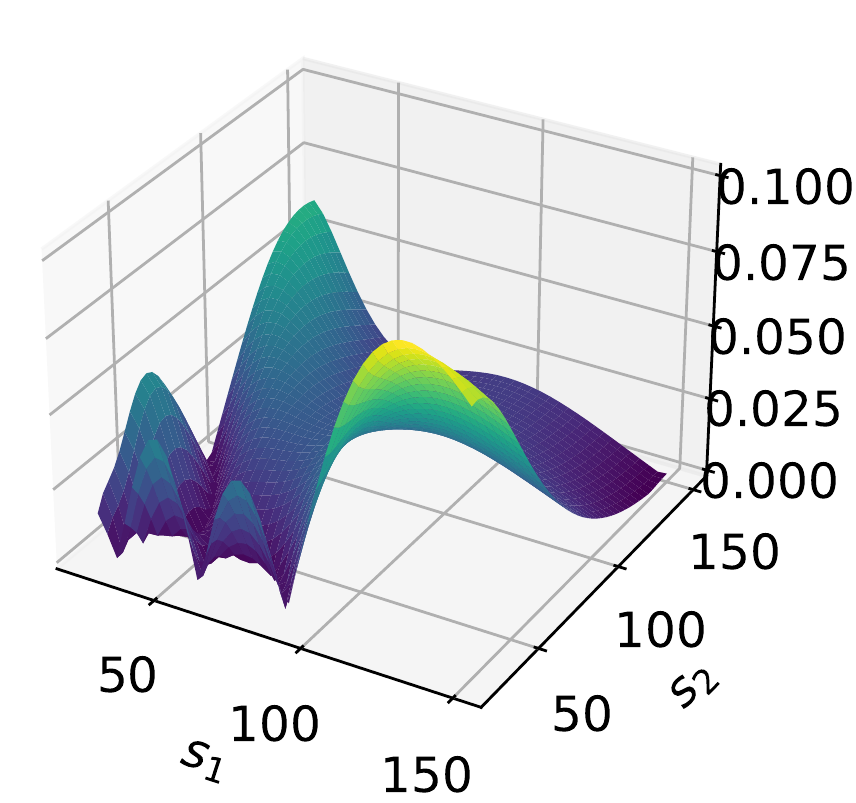}\hspace{1em}
\includegraphics[width=.3\textwidth]{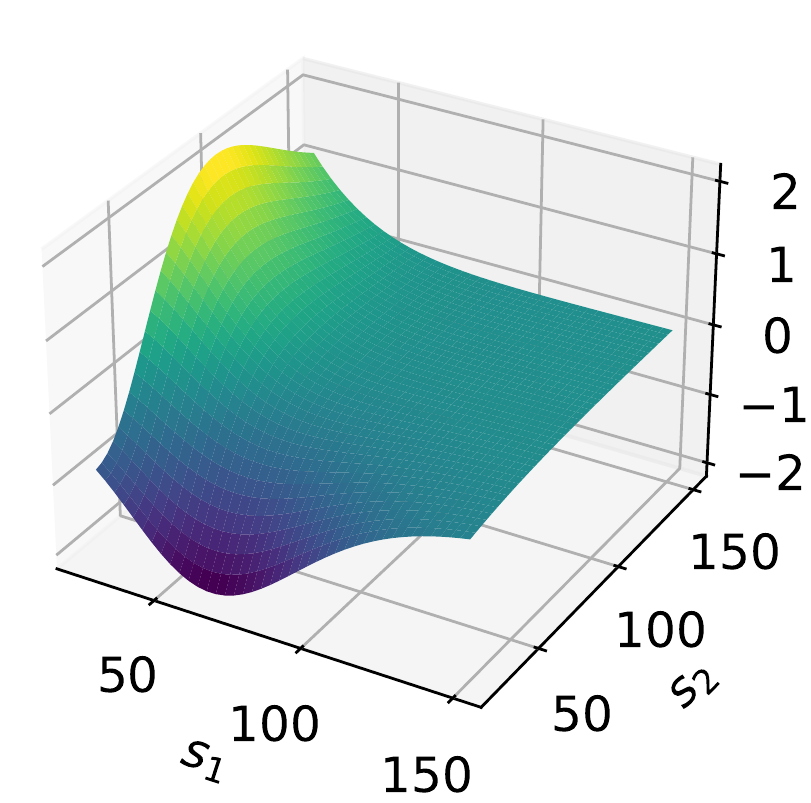}
\caption{From left to right: Deep parametric PDE approximation, errors  and approximated residual   value for two underlyings.
All pictures evaluated for $\sigma_1 = 0.1$ and $\sigma_2 = 0.3$ with $\rho_{12} = 0.5$ at maximal time to maturity $t=4$.
}
\label{fig:solution_numerical_and_error_2d}
\end{center}
\end{figure}

\begin{figure}[htbp]
\begin{center}
\begin{subfigure}{.49\textwidth}
\begin{center}
\includegraphics[width=.49\textwidth]{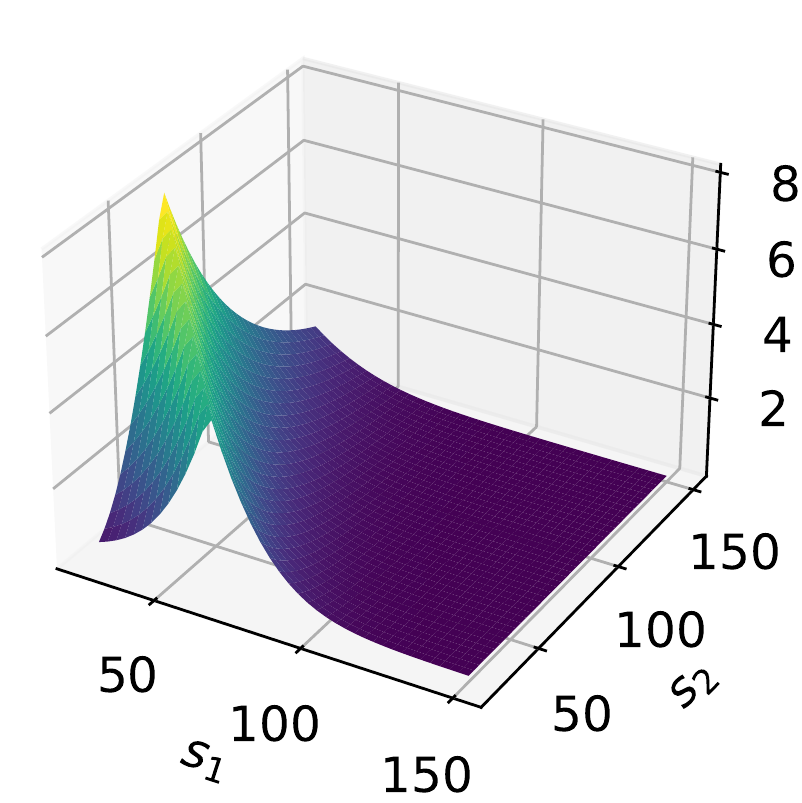}
\includegraphics[width=.49\textwidth]{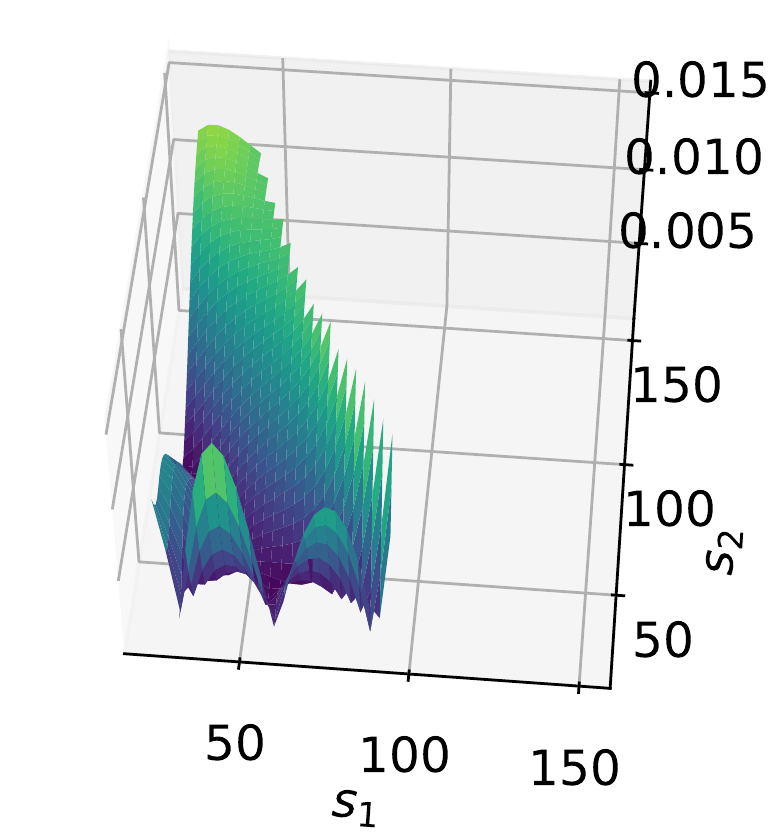}\hfill
\caption{$\sigma_1 = 0.1$ and $\sigma_2 = 0.3$.}
\label{fig:iv_2d_case_a}
\end{center}
\end{subfigure}\hfill
\begin{subfigure}{.49\textwidth}
\begin{center}\hfill
\includegraphics[width=.49\textwidth]{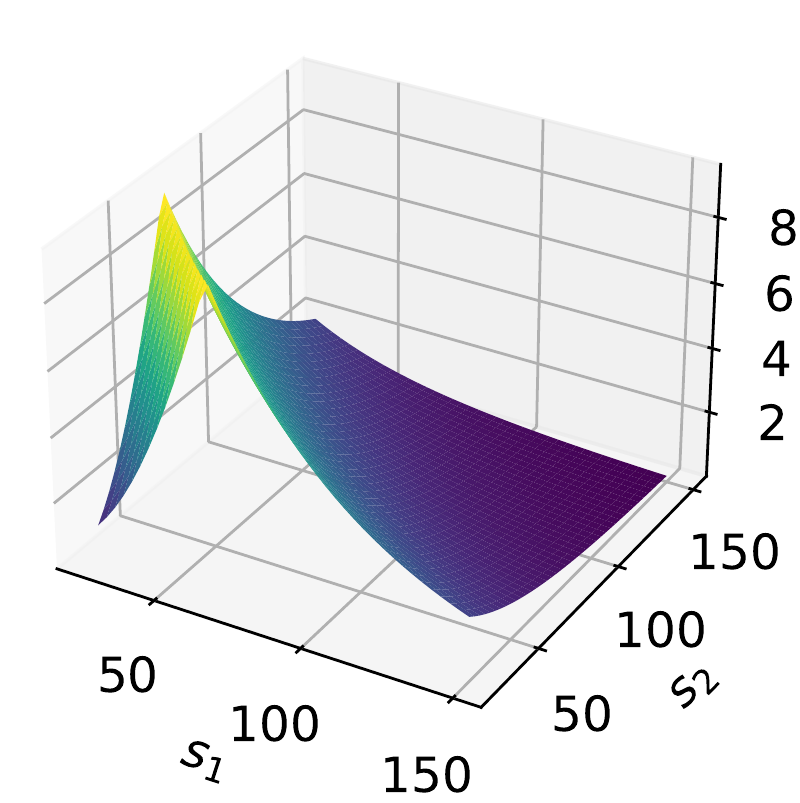}
\hfill
\includegraphics[width=.49\textwidth]{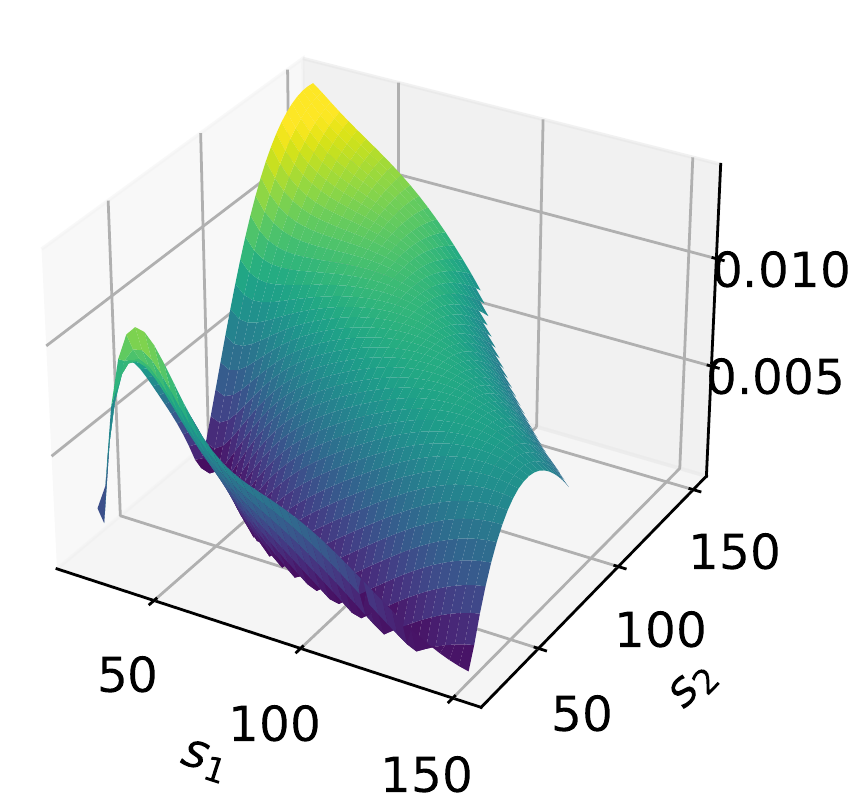}
\caption{$\sigma_1 = \sigma_2 = 0.3$}
\label{fig:iv_2d_case_b}
\end{center}
\end{subfigure}
\caption{
Left: Difference between the exact option price and the lower trivial no-arbitrage bound $c_{\rm{lb}}$. Right: relative error of the implied volatility, computed where the difference is larger than $0.5$.}
\end{center}
\end{figure}

For a more detailed evaluation of the solution quality, we also consider the implied volatility. As outlined in Section~\ref{sec:implied_volatility}, we use a concept extending the standard implied volatility to the multi-asset case.  
When options are far in the money or out of the money, the implied volatility becomes extremely sensitive. Therefore, we consider implied volatilities on a smaller domain, excluding these extreme cases.  In the left of Figure~\ref{fig:iv_2d_case_a}, we see the difference between the exact option price and the trivial no-arbitrage bound $c_{\rm{lb}}$. The difference is small far in- and far out of the money. We consider the implied volatility  as an interesting measure, where the difference  is larger than a threshold, which we choose as $0.5$.  The relative error of the implied volatility in this subdomain shown is on the right of Figure~\ref{fig:iv_2d_case_a}. Figure~\ref{fig:iv_2d_case_b} shows the same quantities for different parameter values. Thanks to the parametric approach, in both situations the solution is computed by evaluating the same neural network. With the second set of parameters, a smaller portion of the domain is far in the money. With both sets of parameters, the relative error in the implied volatility is smaller than $1.5\%$ and significantly smaller in most of the domain of interest.

For three underlyings, we   consider three-dimensional plots, varying each of the asset prices at final time and for fixed parameter values. In Figure~\ref{fig:errors_3d}, we show an excerpt of the error for $\sigma_1=\sigma_3=0.1$, $\sigma_2=0.2$, $\rho_{1,2}=0.2$ and $\rho_{2,3}=0.5$.
As seen in the left of the figure, in most parts of the domain of interest, the pricing error is less than $0.075$. Only a small  part of the domain exhibits errors of up to $0.22$. 
\begin{figure}[htbp]
\begin{center}
\includegraphics[width=.3\textwidth]{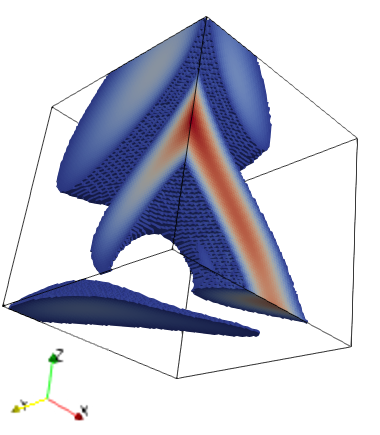}\hspace{1em}
\includegraphics[height=.3\textwidth]{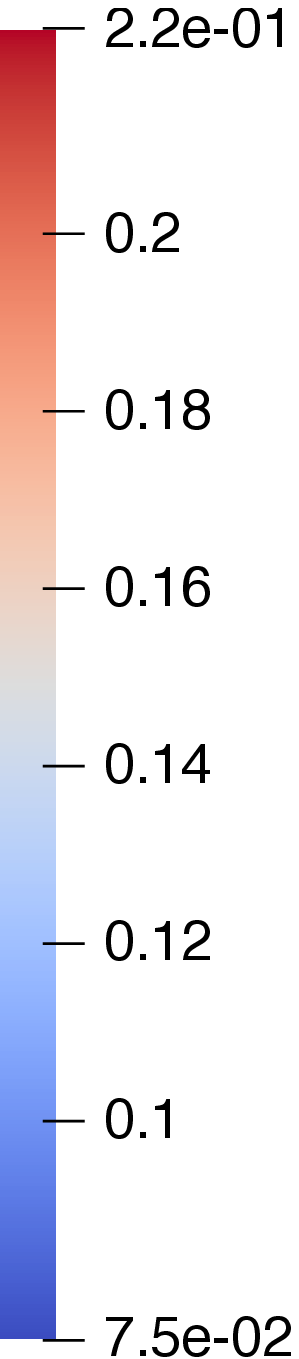}
\hspace*{.1\textwidth}
\includegraphics[width=.3\textwidth]{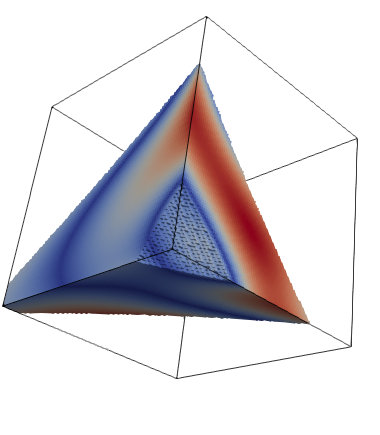}
\hspace{1em}
\includegraphics[height=.3\textwidth]{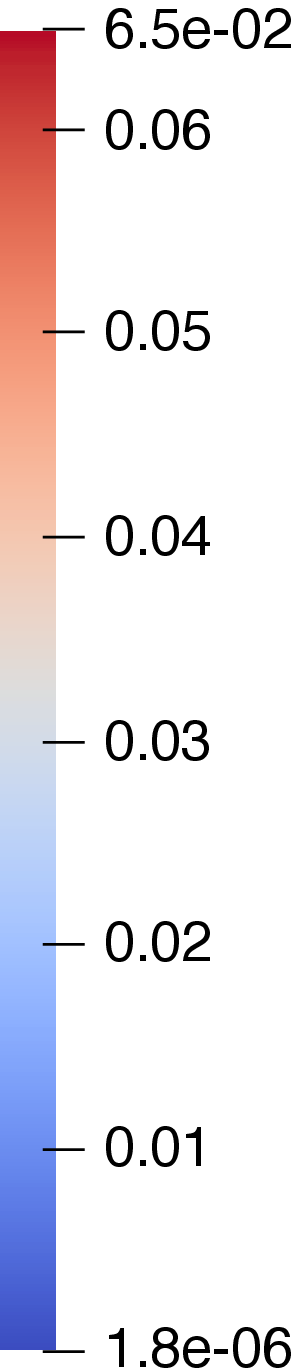}
\caption{Different errors for three underlyings.
Left: Pricing error, only shown where it is larger than $0.075$.
Right: Relative error in the volatility, evaluated where the difference to $c_{\rm lb}$ is larger than $0.5$.
Both images share the same perspective, facing the corner where $S_1=S_2=S_3=25$. The  parameters are $\sigma_1=\sigma_3=0.1$, $\sigma_2=0.2$, $ \rho_{1,2}=0.2$ and $ \rho_{2,3}=0.5$. }
\label{fig:errors_3d}
\end{center}
\end{figure}
In the right of the figure,  we show the relative error of the implied volatility. We see an error well below $10\%$ throughout the domain with  parts below $1\%$.

In the cases of more than three underlyings, we cannot display the $d$ asset prices any more. Instead, we sample over values with the same average asset price. In Figure~\ref{fig:solution_numerical_and_error_higher_dim}, we display the average option price and the maximal error over a set of randomly selected asset prices with the same average for the case five and eight underlyings. We see a similar behaviour as in the lower dimensional cases. The maximal error remains well below $0.2$. 
Note that with eight underlyings, there are $16$ parameters and the total dimensionality of the problem is $25$.

\begin{figure}[htbp]
\begin{center}
\includegraphics[width=.3\textwidth]{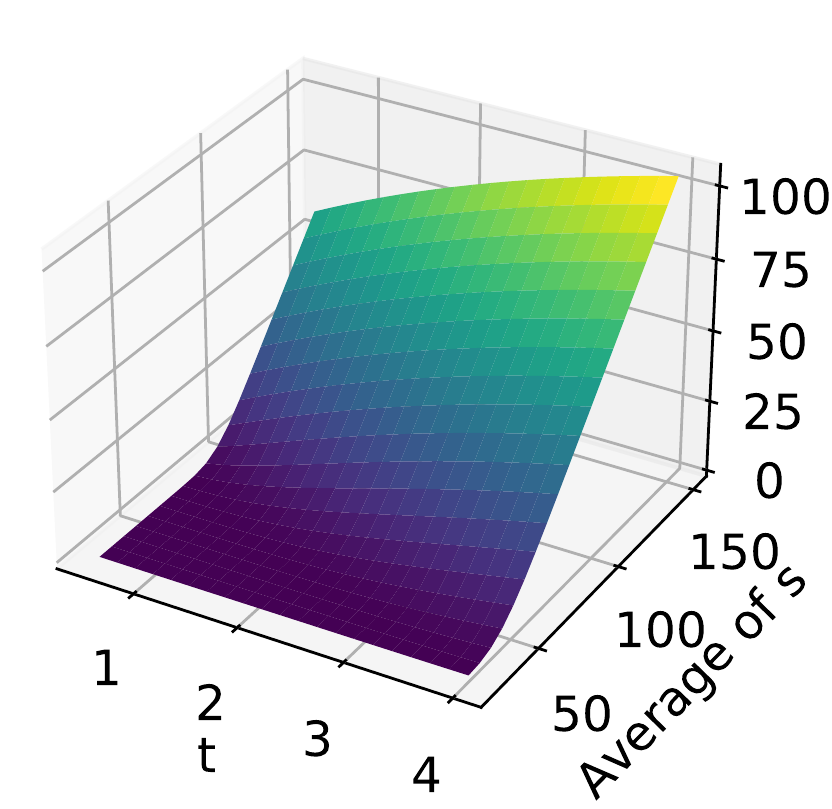}\hspace{1em}
\includegraphics[width=.3\textwidth]{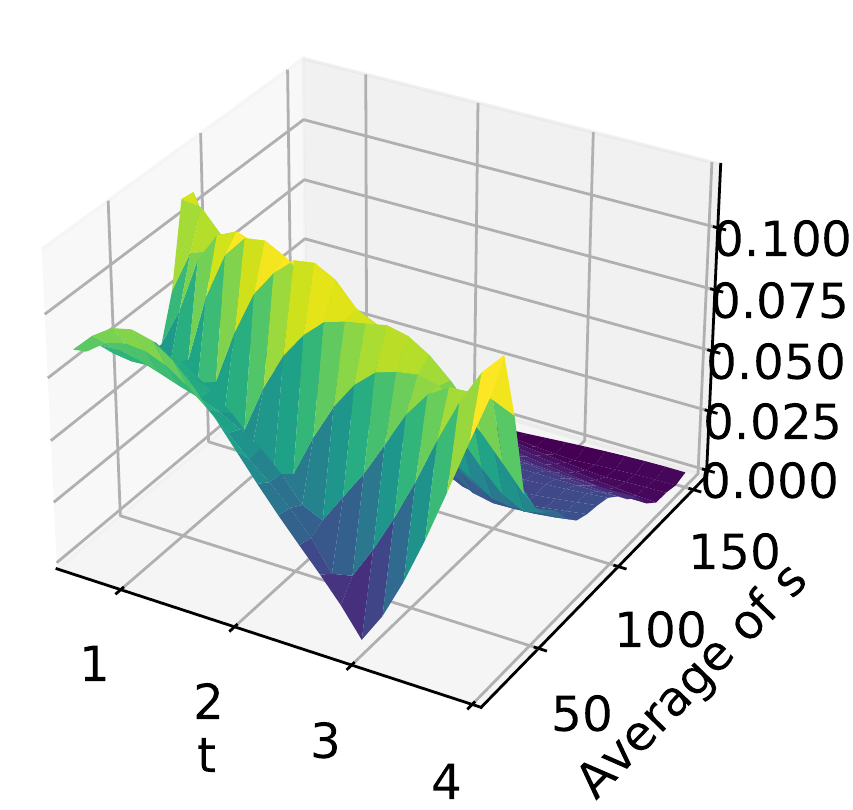}
\\
\includegraphics[width=.3\textwidth]{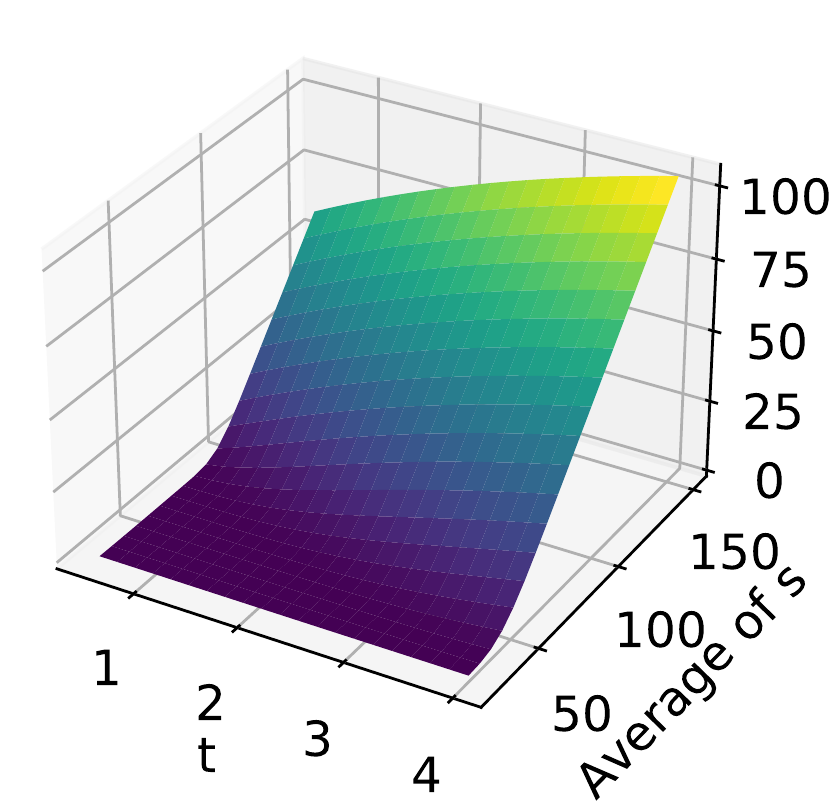}\hspace{1em}
\includegraphics[width=.3\textwidth]{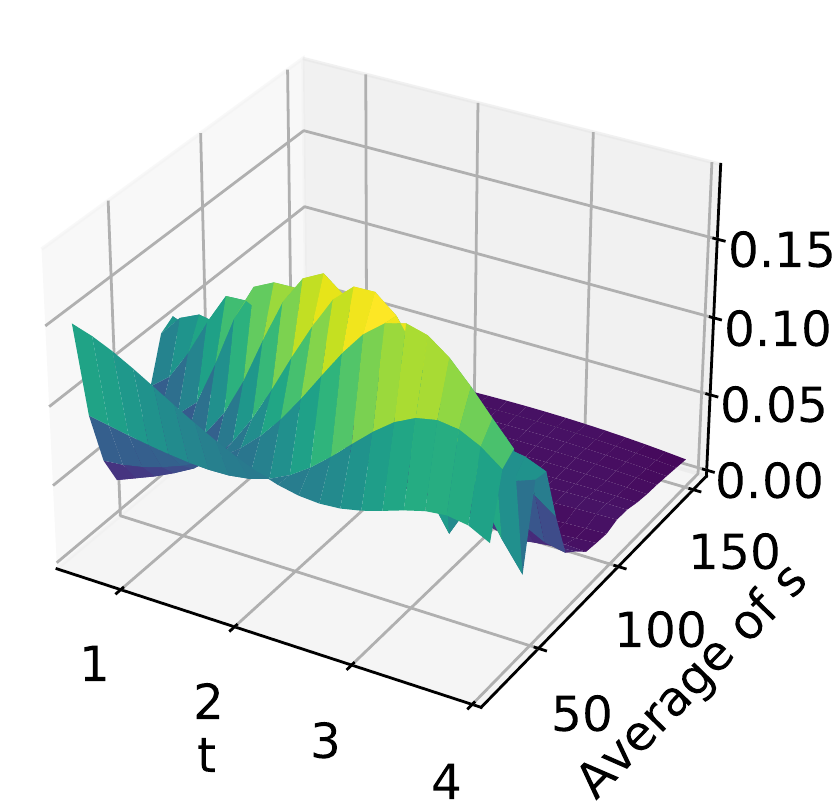} 
\\
\caption{Averaged solution (left) and maximal errors (right) for five (top) and eight (bottom) underlyings. }
\label{fig:solution_numerical_and_error_higher_dim}
\end{center}
\end{figure}

\subsubsection{Evaluation on the Whole Parameter Domain}
So far, we have discussed the approximation error for arbitrary but fixed parameters. As the deep parametric PDE method is able to approximate the solution on the whole parameter domain, we now consider errors for varying parameters.

\begin{figure}[htbp]
\begin{center}
\includegraphics[width=.30\textwidth]{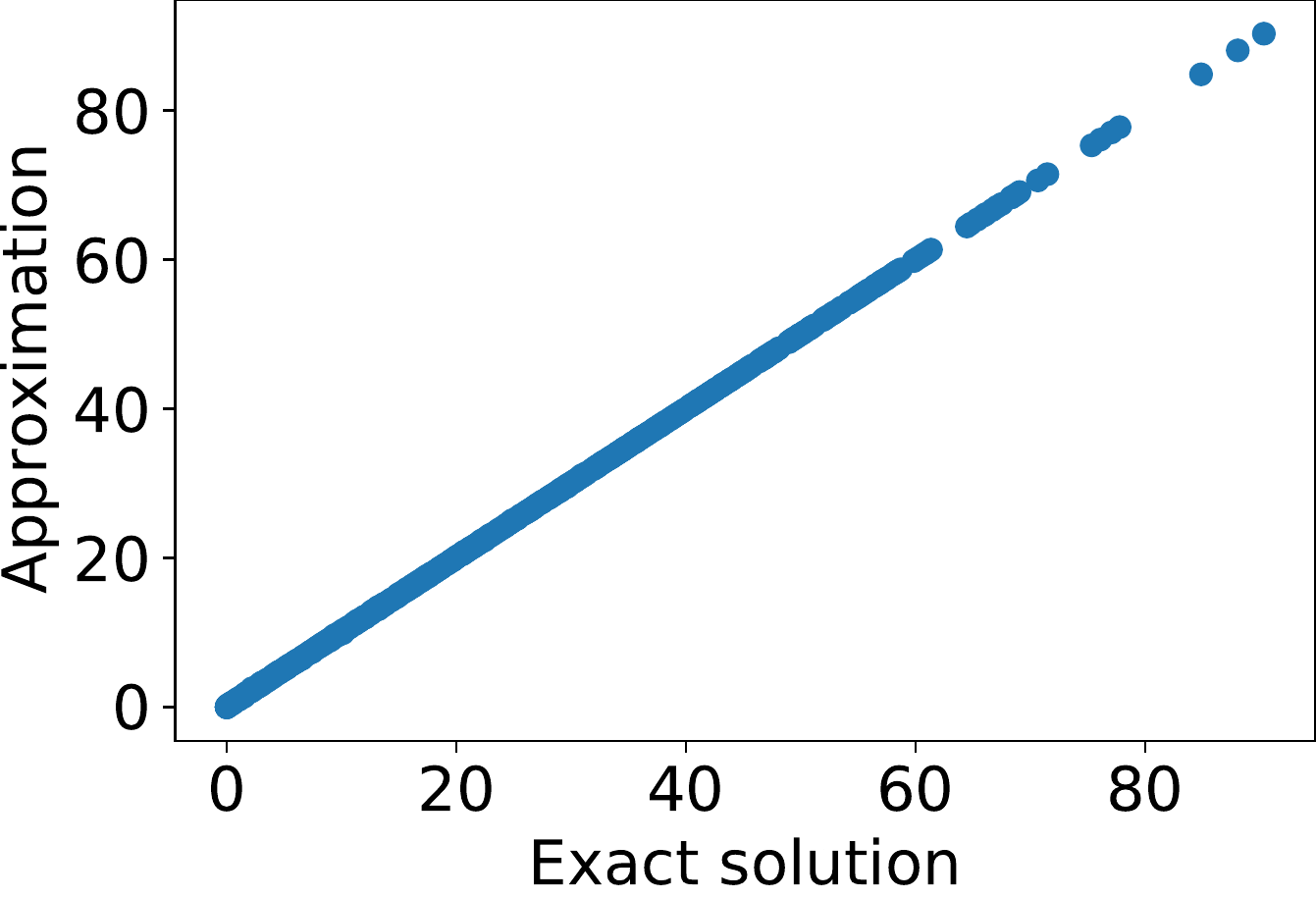}\hspace{1em}
\includegraphics[width=.30\textwidth]{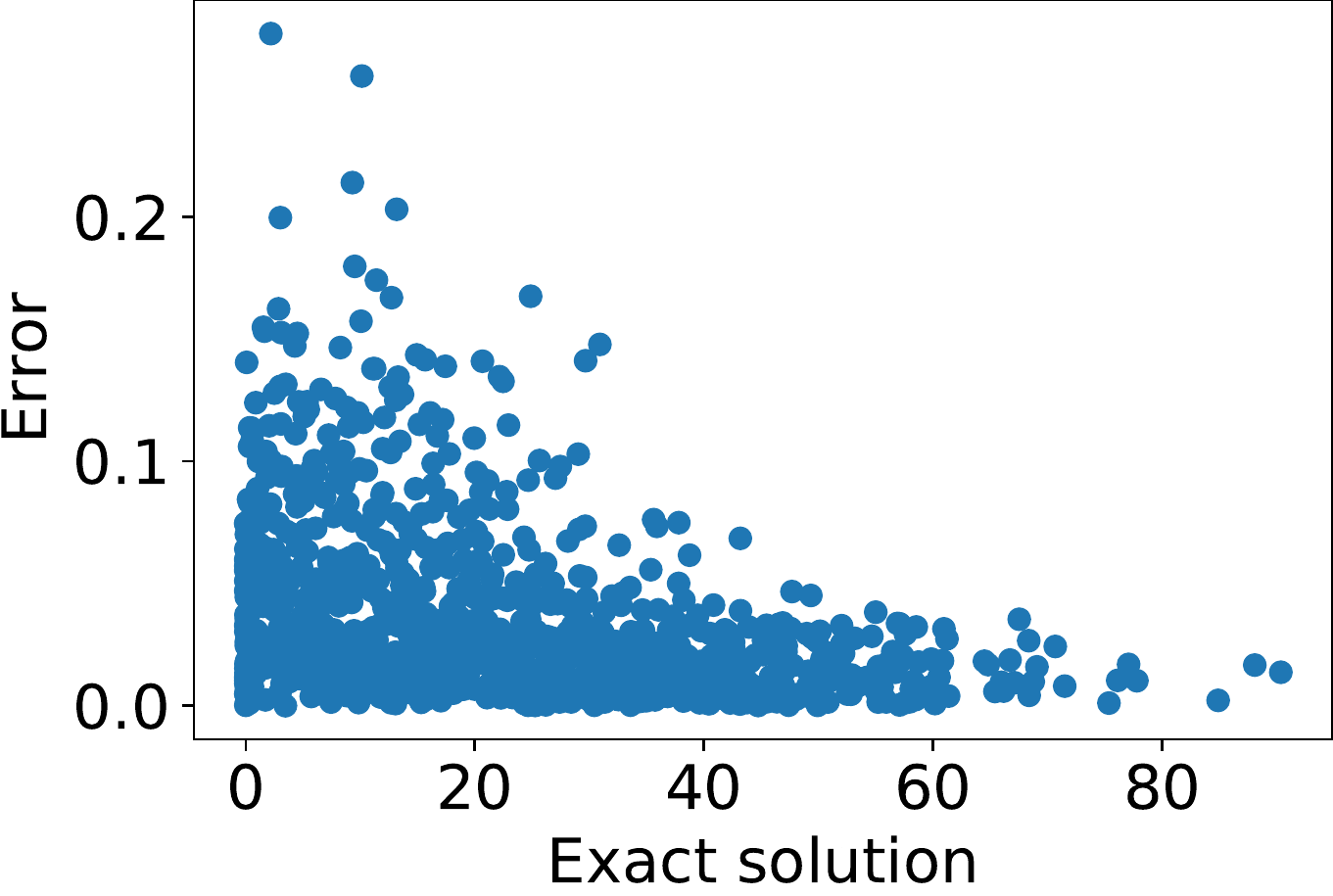}
\hspace{1em}
\includegraphics[width=.30\textwidth]{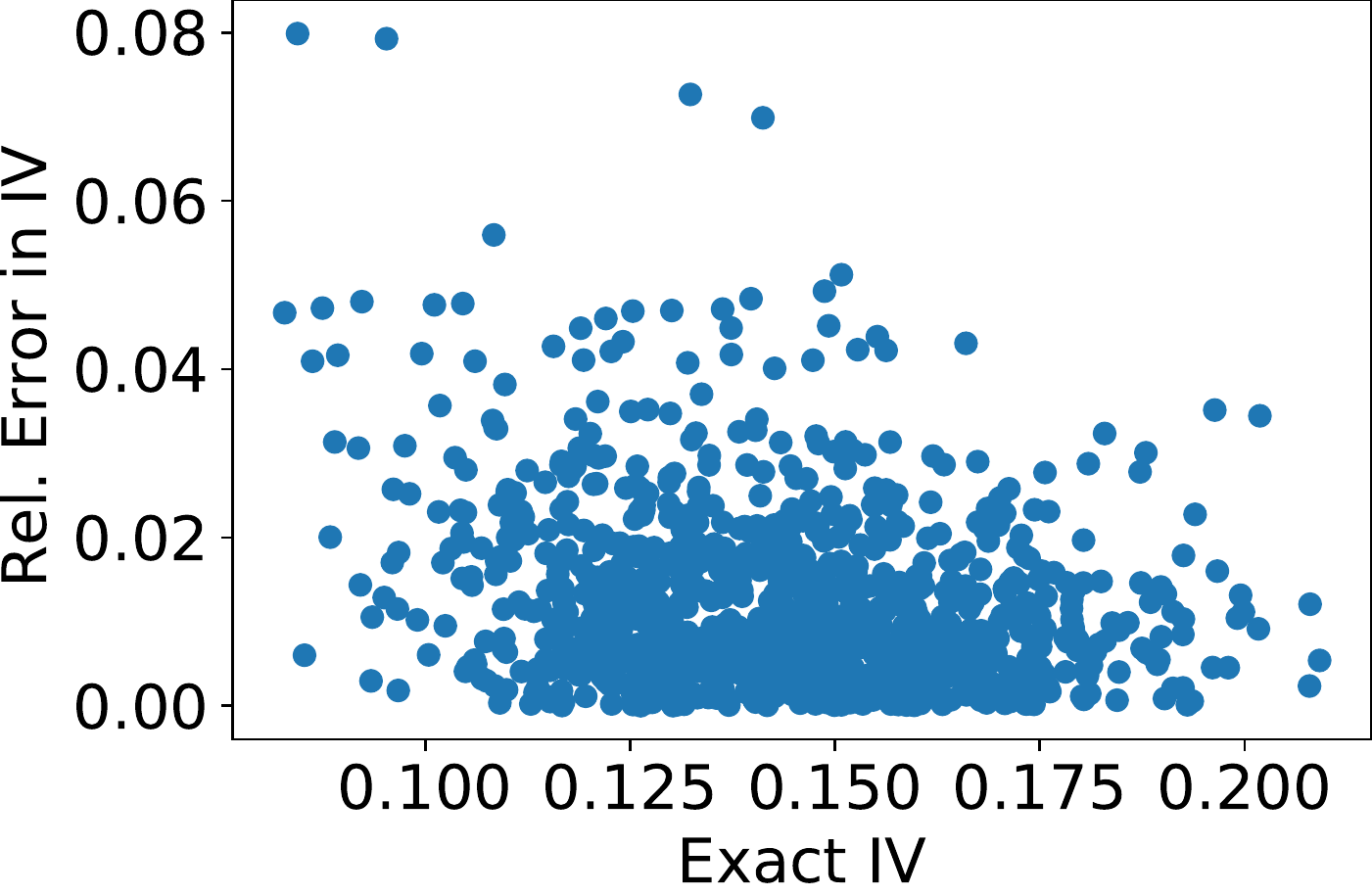}\\
\includegraphics[width=.30\textwidth]{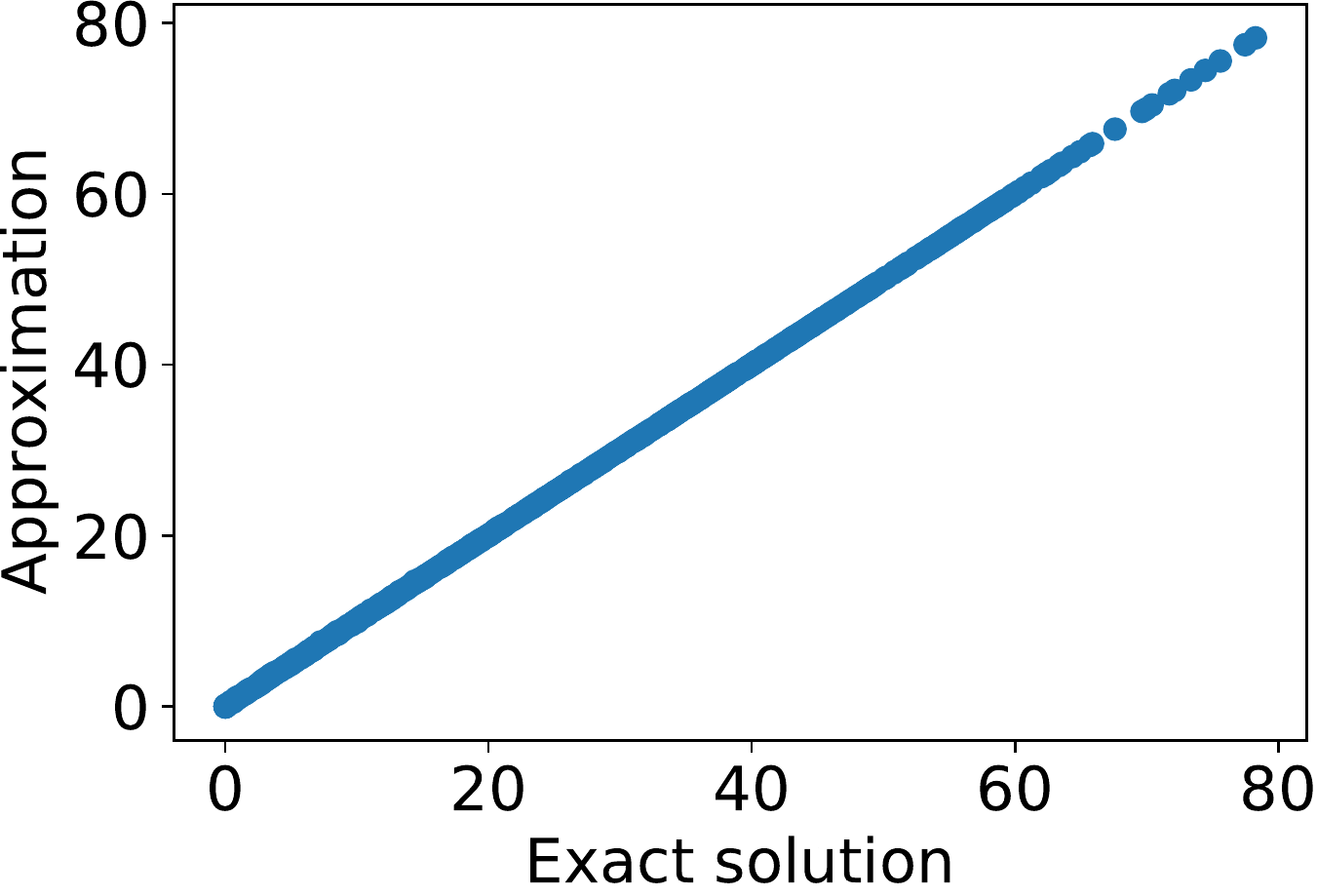}
\hspace{1em}
\includegraphics[width=.30\textwidth]{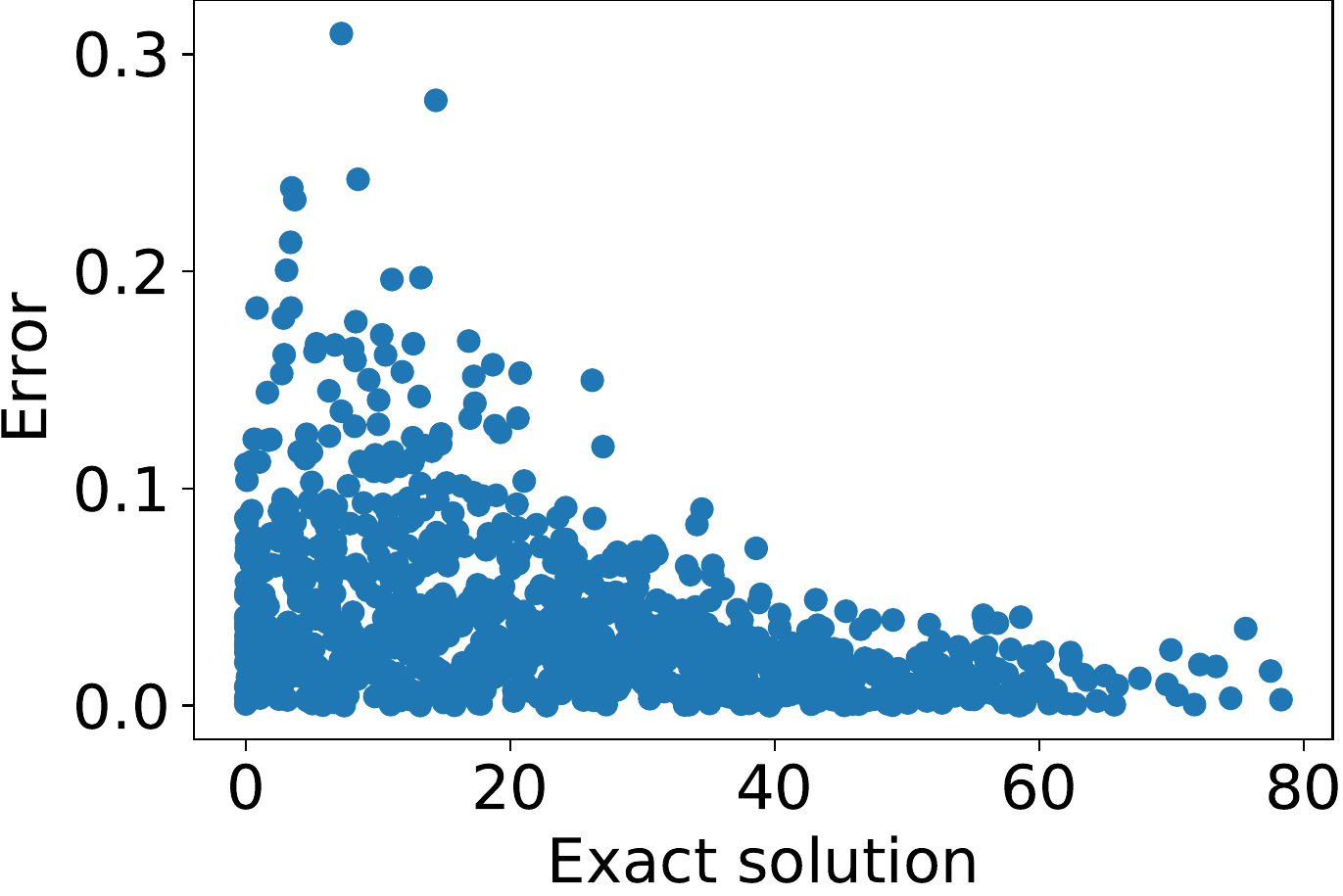}\hspace{1em}
\includegraphics[width=.30\textwidth]{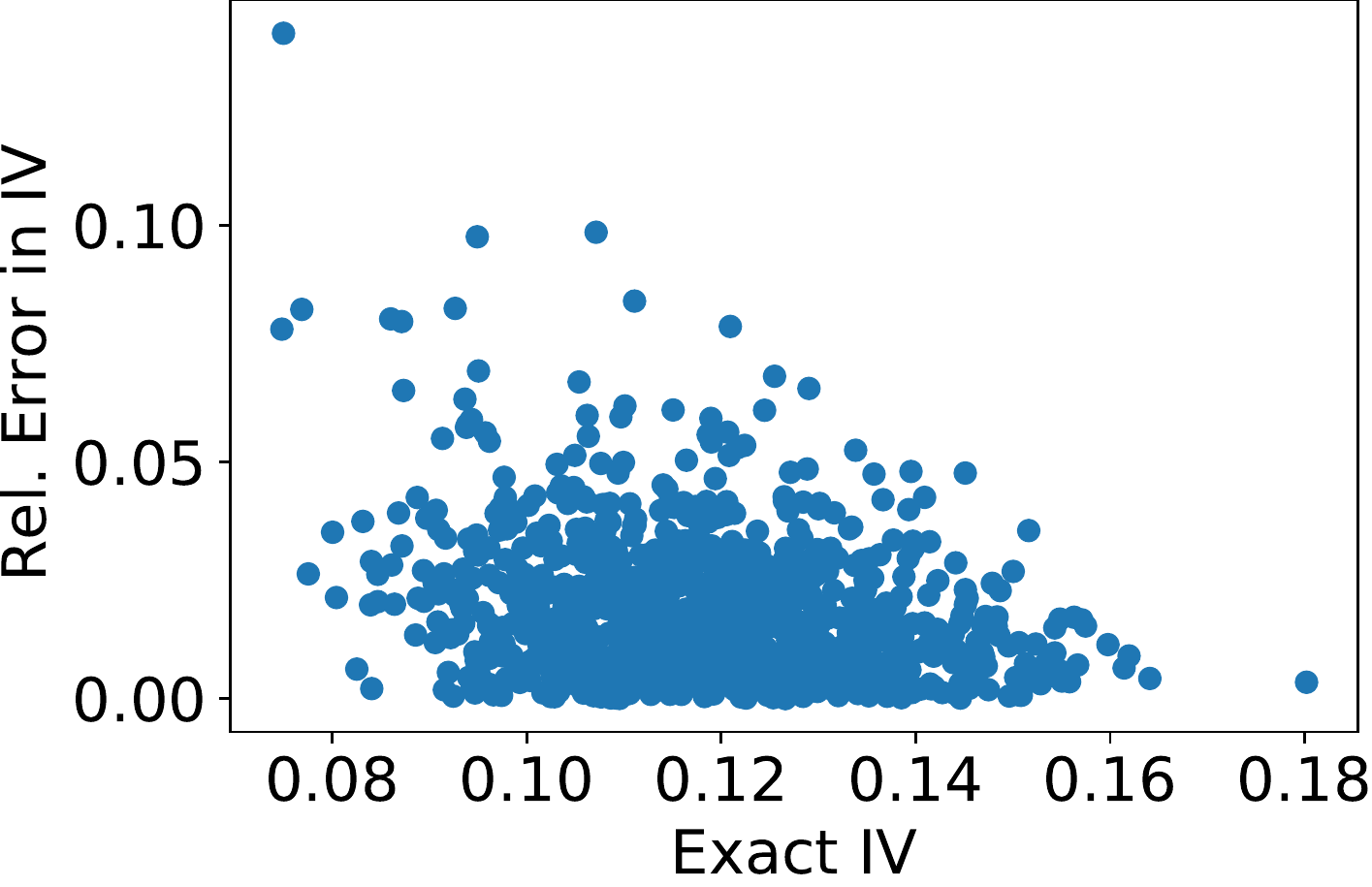}
\caption{Scatter of prices  evaluated for $1,000$ different tuples of time, asset prices and parameters chosen uniformly at random   from the region of interest. Top: Five underlyings. Bottom:  Eight underlyings. Left: Exact solution and approximation on the x- and the y-axis. Points close to the diagonal show a small error. Middle: Plot of the error against the exact solution. Right: Relative error of the implied volatility (IV), computed for option prices with a difference to $c_{\rm lb}$ of at least $0.5$. }
\label{fig:scattered_solution_and_error_higher_dim}
\end{center}
\end{figure}
The parametric solution has up to $25$ dimensions, posing challenges to visualising the error. As a first test, we consider $1,000$ random points from the whole domain of interest. These points have different time, state and parameter values, but with the deep parametric PDE method the option prices are quickly evaluated using a single neural network. In Figure~\ref{fig:scattered_solution_and_error_higher_dim}, different scatter plots are shown, visualising the exact and approximated solution, the absolute error and the relative error in the implied volatility. Most absolute errors remain below $0.1$. Isolated points can be seen with larger errors of up to $0.3$.  For the implied volatility a similar picture as before is seen, with most relative error values below $5\%$  and all relative error values below $15\%$. 
As the overall approximation is good, but there are some outliers, we further investigate the parameter-dependency of the error for the different dimensionalities of the problem.

\begin{figure}[htbp]
\begin{center}
\includegraphics[width=.3\textwidth]{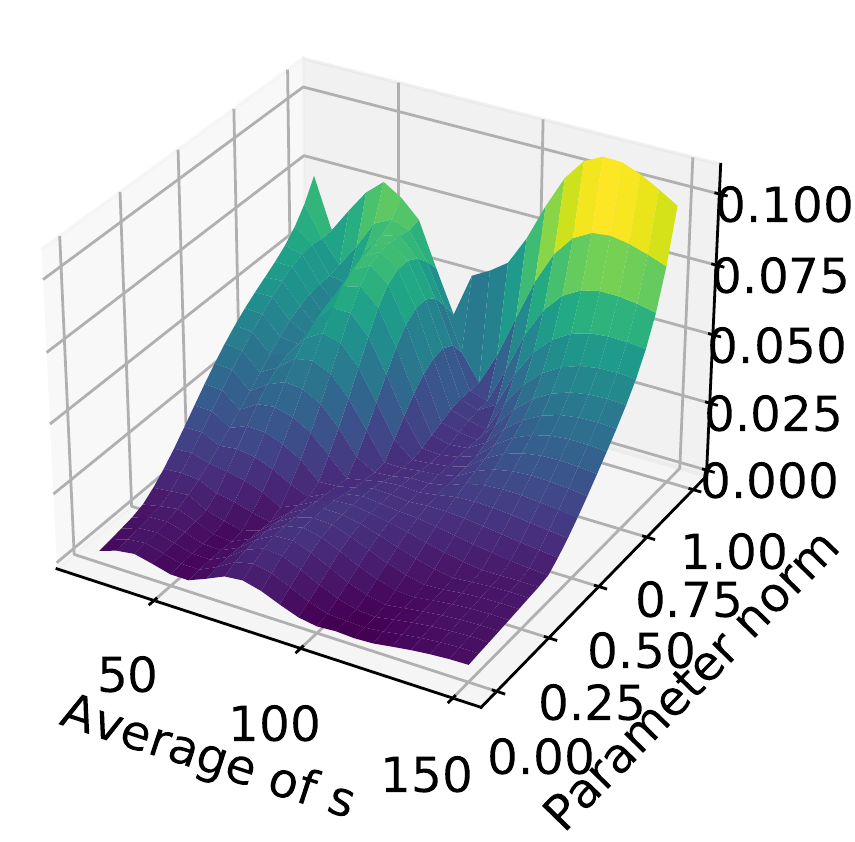}
\hfill
\includegraphics[width=.3\textwidth]{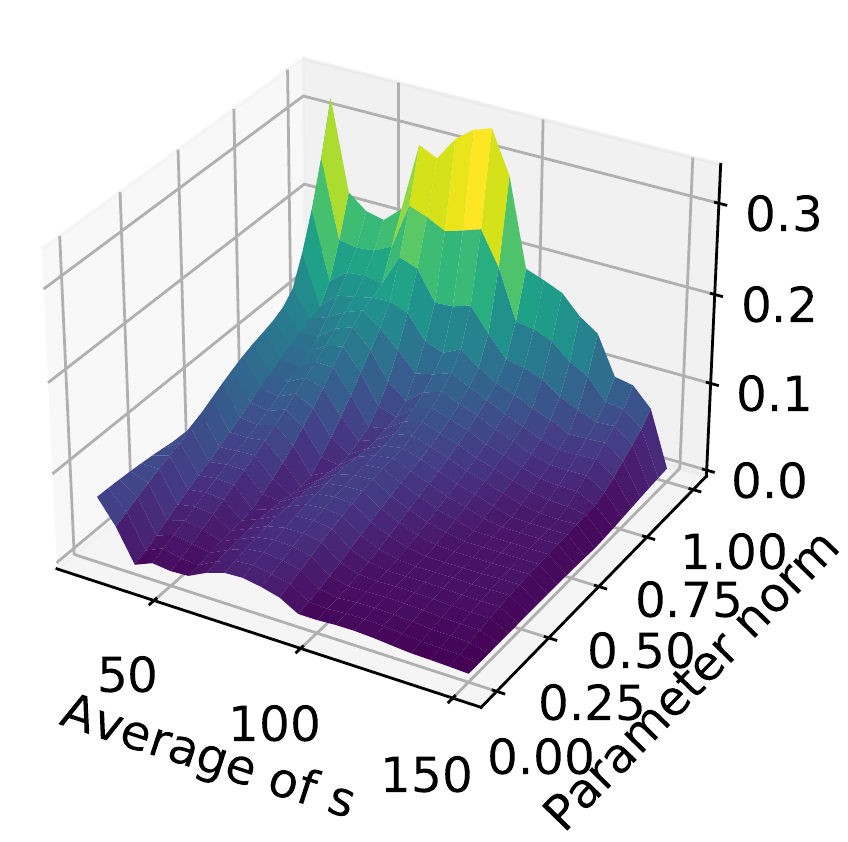}
\hfill
\includegraphics[width=.3\textwidth]{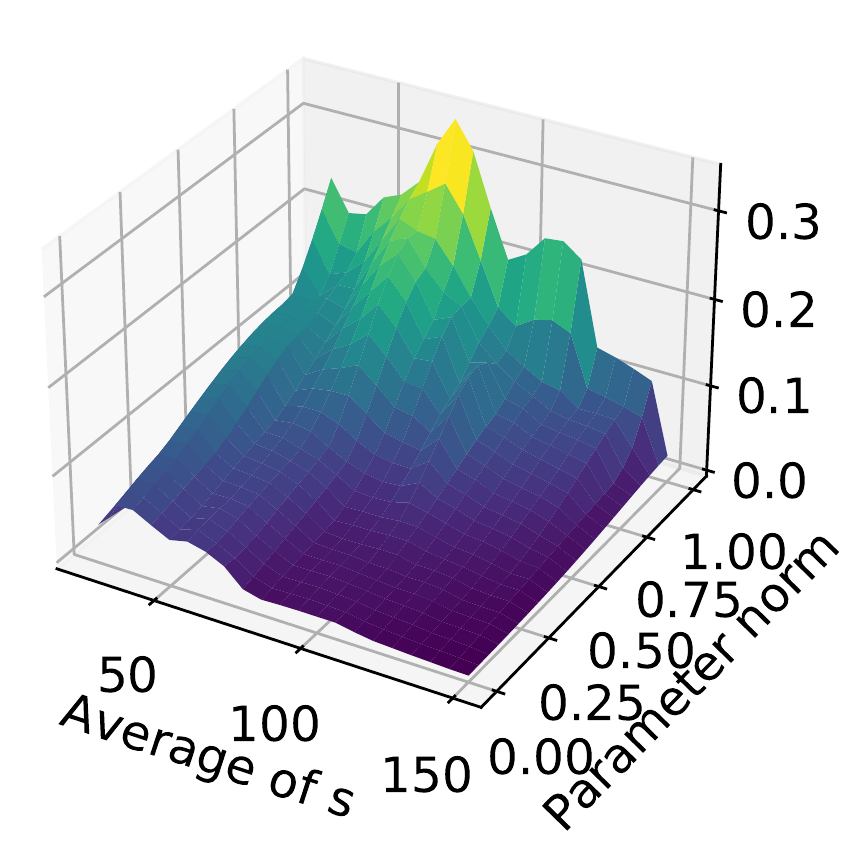}
\\
\includegraphics[width=.3\textwidth]{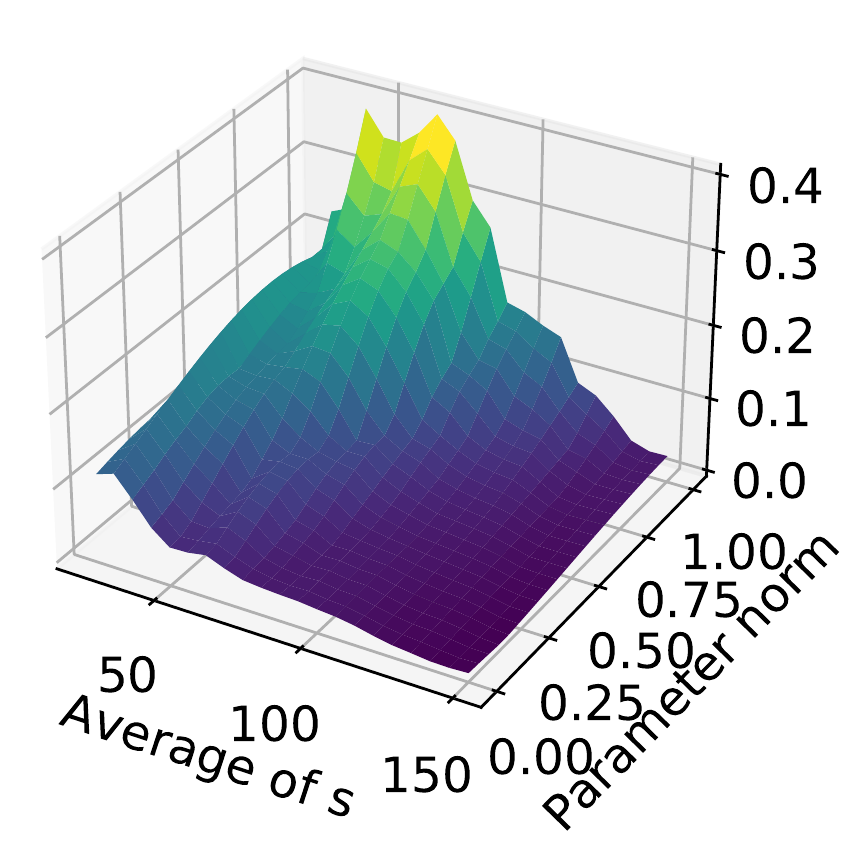}
\hspace{1em}
\includegraphics[width=.3\textwidth]{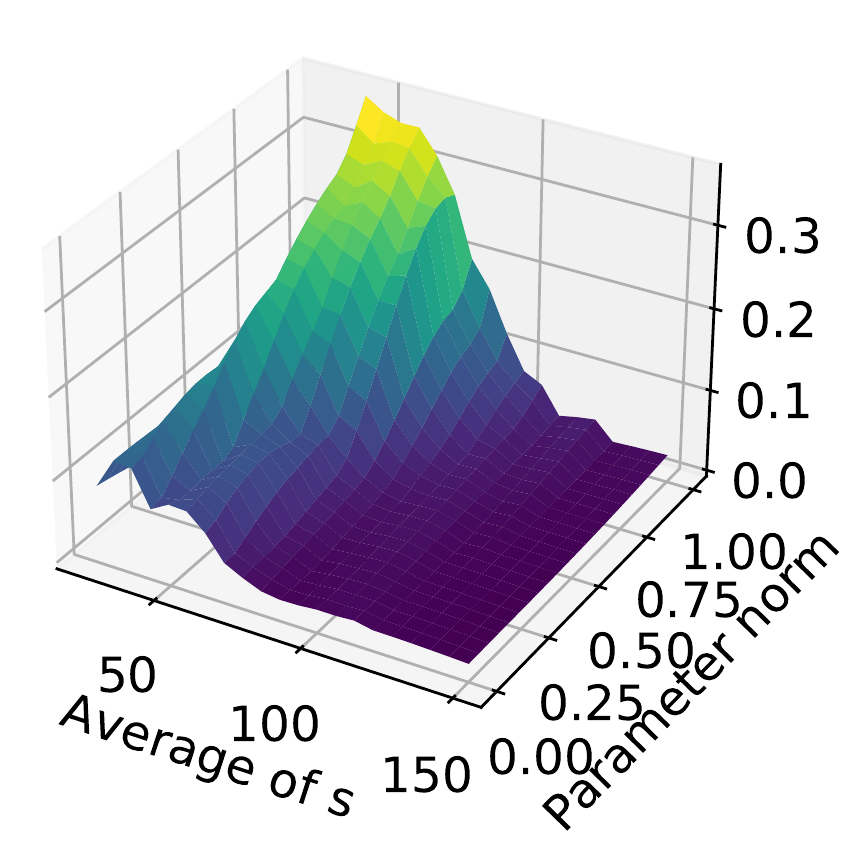}
\caption{Maximal error depending on the parameter and the average asset price at maximal time to maturity. Top row from from left to right $d=1, 2, 3$. Bottom row $d=5,8$.}
\label{fig:max_error_at_final_time}
\end{center}
\end{figure}

In Figure~\ref{fig:max_error_at_final_time}, the maximal error over different samples with the same mean asset price  and parameter norm are shown at maximal time to maturity. We note that the parameter values are normalised, i.e., $\Param=[-1,1]^{n_\mu}$, which means that the parameter norm measures the distance to the reference parameters $r = 0.2$, $\sigma_i = 0.2$ and $\hat\rho_i = 0.5$. We consider the maximal distance, i.e., the $\ell^\infty$-norm: $\|\mu\|_\infty=\max_{i\in\{1,\ldots, n_\mu\}} |\mu_i|$.  The mean asset price is given as $\bar S = \sum_{i=1}^d S_i/d$.
The offline runtimes for the cases $d=1,2,3,5,8$ are $6, 10, 15, 17$ and $58$min. While runtimes grow with the dimension, they clearly do not exhibit a curse of dimensionality.

For each combination of $\bar S$ and $\|\mu\|_\infty$,  $10,000$ samples are considered and the maximal error is shown in the figure. For eight underlyings, evaluating the reference pricer becomes very costly, so only $1,000$ samples are used. We see a clear increase of the error with the parameter norm, indicating that for parameters close to the center of the domain we have a higher accuracy. In the lower-dimensional cases, we see only little dependency on the average asset price, but the influence increases with more assets. In these cases, we see a peak  of the error when the option is near the money. However, as seen in the previous section, this does not necessarily imply an increase in the relative error when considering the implied volatility. 
We note that the maximal error does not increase with the dimension of the problem, but remains stable for the considered basket options with a total dimension of up to $25$.

\subsubsection{Comparison to Alternative Machine Learning Methods}
To fully understand the performance of the deep parametric PDE method, we comapare it to two other machine learning method: the deep BSDE solver~\cite{han:jentzen:18} and the deep Galerkin method~\cite{sirignano:17}. 

To compare with the deep BSDE solver, we use the code provided by the authors of~\cite{han:jentzen:18} on GitHub.  We adapt the example of the Black-Scholes equation with default risk to fit our setting. The only required changes are to change the payoff to a European basket call, to set the default risk to zero and to choose the number of underlying assets as $8$. We keep the parameters  as $T = 1$, $s_i = 100$, $r = 0.02$, $\sigma_i = 0.2$ and $\rho_{ij}=0$ for $i,j=1,\ldots, 8$, $i\neq j$. Within $6,000$ iterations the method does not converge yet, so we choose $15,000$ iteration steps.
For a fixed triple of time, state and parameters $(\hat \timet,\hat x ;\hat \mu)$, the deep BSDE method trains a neural network to compute the single option price $u(\hat \timet, \hat x;\hat\mu)$.

As the specified parameters are not part of our previously used parameter set, we retrain the deep parametric PDE method on a slightly different domain with  $r\in (0, 0.1)$ and $\hat\rho \in (-0.2, 0.2)$. 

The deep Galerkin method is related to the deep parametric PDE method. It shares the least-squares formulation of the PDE, but trains a single solution $(\timet, x)\mapsto u(\timet,x;\hat\mu)$ in the time-state-space for a fixed  parameter $\hat\mu$.
As an implementation of the deep Galerkin method, we therefore adapt the code of the deep parametric PDE method accordingly.
The original deep Galerkin method as presented in~\cite{sirignano:17}, does not include a localisation which improves the accuracy. In order to examine the effect of solving for a whole parameter domain against solving for a fixed parameter set, we keep all other specifications equal.  Particularly, we also use the localisation for the deep Galerkin method. 

All values and the reference pricer are listed in Table~\ref{tab:comparison_deepBSDE} and compared to a Monte-Carlo solution with $10^{9}$ samples. We see similar relative errors of less than $1\%$ for all machine learning pricers, with the deep parametric PDE method slightly closer to the true solution. It is worth noting that also the reference pricer is only a bit more accurate, which confirms the efficiency of deep neural networks.  

Table~\ref{tab:comparison_deepBSDE} also includes a comparison of runtimes. 
Reported runtimes for the offline-phases were measured on a GPU node on Queen Mary's cluster Apocrita, using an Nvidia Tesla V100. Online runtimes are measured on the end user device.  
Once the training is finished, the deep parametric PDE method provides the fastest solution to evaluate option prices with different parameters. The speed-up factor compared to the second fastest method, which is the reference pricer, is over  $30$. 
The evaluation time and accuracy of the deep Galerkin method and the deep parametric PDE method are both equivalent. For the cost of doubling the offline runtime, the deep parametric PDE method delivers the solution for the whole parameter set with no loss of accuracy.
Already calling the solver for three different parameter sets, the deep parametric PDE method yields a total runtime gain of roughly half an hour.

While the offline run-time of the deep parametric PDE method is close to an hour, it only needs to be performed once. This can be done at idle times, i.e., at night.  Afterwards, when the solution is required in real-time  it is readily available.

\begin{table}
{\small
\begin{tabular}{@{}p{4em} p{5.9em} p{5.9em} p{5.9em} p{5.9em}p{5.9em}@{}}
\toprule
& Monte-Carlo & Reference pricer & Deep BSDE & Deep Galerkin  & Deep Parametric PDE \\
\midrule
 value & 3.9166&3.9217&3.8986&3.9335&3.9327\\\addlinespace
rel. error & -- & 0.130\%  &0.460\%  & 0.431\%&0.411\% \\ \addlinespace
offline runtime  &--&--&--& {28min \newline  for each $\hat \mu$ }& {59min \newline once $\forall \mu\in\Param$} \\\addlinespace
online  runtime &6min 41s&1.34s&12min 7s&39.9ms&41.9ms\\	
\bottomrule
\end{tabular}
}
\caption{Comparison of different methods to compute the value of an at the money European basket call option in the Black-Scholes model.  Strike price $100$ with eight underlying assets $s_i=100$, $\sigma_i = 0.2$, $\rho_{ij} = 0$, $i, j=1,\ldots, 8, i\neq j$ and $r=0.02$.
Monte-Carlo solution with $10^9$ samples as the exact solution (estimated standard deviation $0.000344$).
}
\label{tab:comparison_deepBSDE}
\end{table}

\subsection{Benchmark Case with Explicit Solution}
To gain more insight into the performance in different settings, we  the deep parametric PDE method on a second, more academic problem. 
The pricing problem with the geometric payoff, as described in Section~\ref{sec:geometric_payoff}, has an explicit solution, which makes it an interesting example. 

Note that the trivial no-arbitrage bound is different than that of a standard basket option. The expectation of the averaged asset prices is $e^{(r-\beta) \timet} e^{\sum_{i=1}^dx_i/d} $, with 
$
\beta = 
\sigma^2 /2  (1 - d^{-1})   (1 - \rho)
$.
Taking this into account in the localisation~\eqref{eq:localisation}, it reads
\begin{align*}
\localisation(\timet,x; \mu) = \frac{1}{\lambda}\log\left(1+e^{\lambda\left(   e^{-\beta \timet} e^{\sum_{i=1}^d x_i/d} -  K e^{-r\timet} \right)}\right).
\end{align*}

The resulting approximation for two underlyings at maximal time to maturity is shown in Figure~\ref{fig:solution_geometric_2d}. We can see error levels below $0.05$ in the domain of interest. Compared to the basket call options, we see a different profile of the approximative residual   value.

\begin{figure}[htbp]
\begin{center}
\includegraphics[width=.3\textwidth]{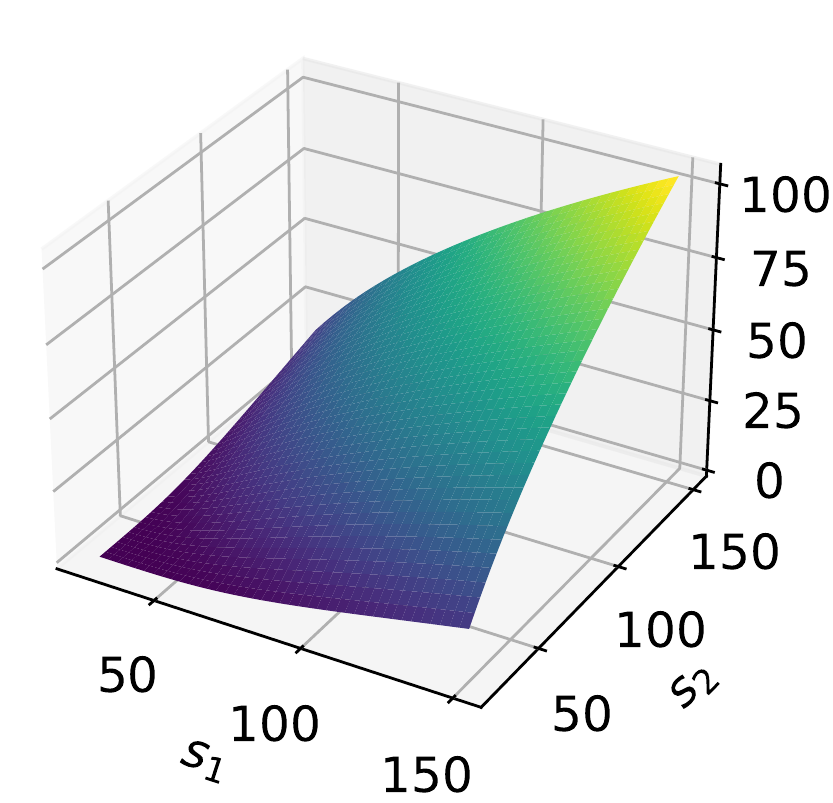}\hspace{1em}
\includegraphics[width=.3\textwidth]{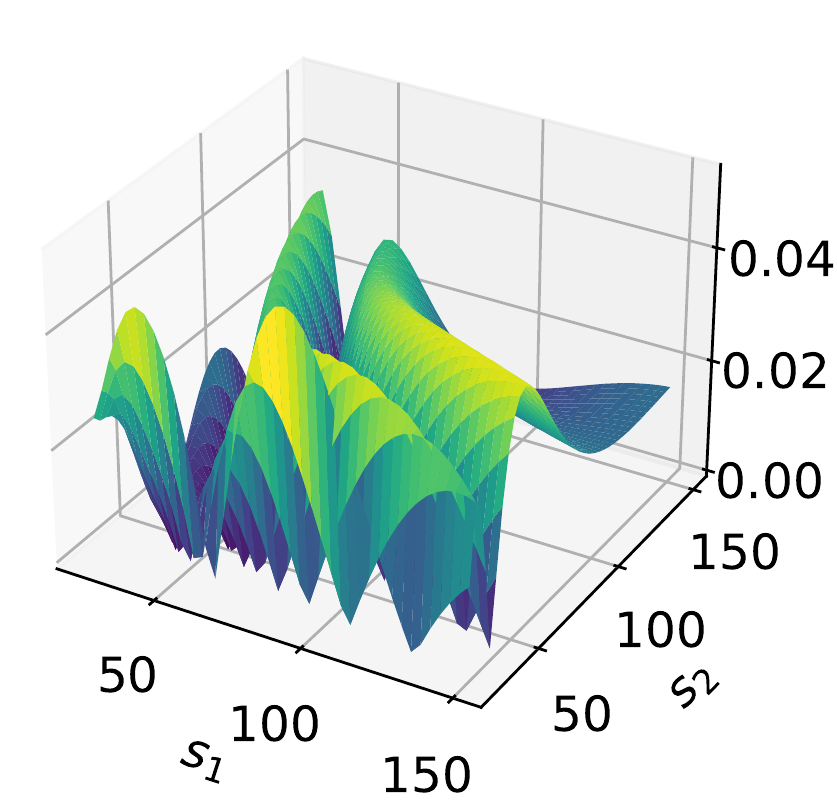}\hspace{1em}
\includegraphics[width=.3\textwidth]{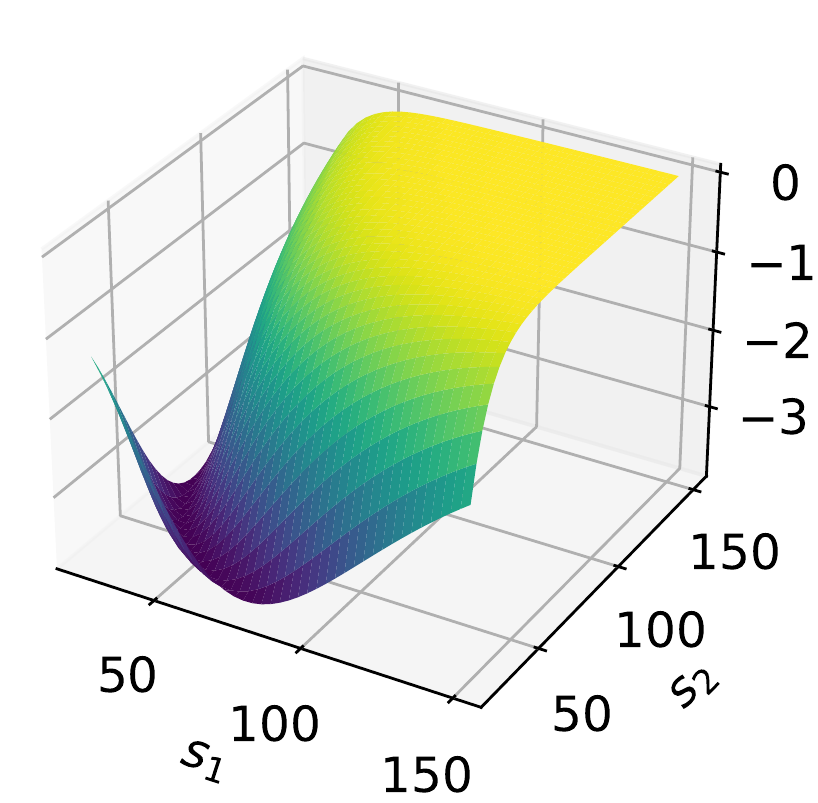}
\caption{From left to right: Deep parametric PDE approximation, errors  and approximated residual   value for two underlyings.
All pictures evaluated for $\sigma = 0.1$ at maximal time to maturity using the geometric payoff.
}
\label{fig:solution_geometric_2d}
\end{center}
\end{figure}
 Figure~\ref{fig:max_error_at_final_time_geometric} shows the maximal error depending on the parameter norm and the average asset price. In all cases, we see a peak of the maximal error near the money and for large parameter values. This means that for applications with a smaller parameter domain of interest, the method is significantly more accurate.
In higher dimensions, the maximal error is only slightly higher.

\begin{figure}[htbp]
\begin{center}
\includegraphics[width=.3\textwidth]{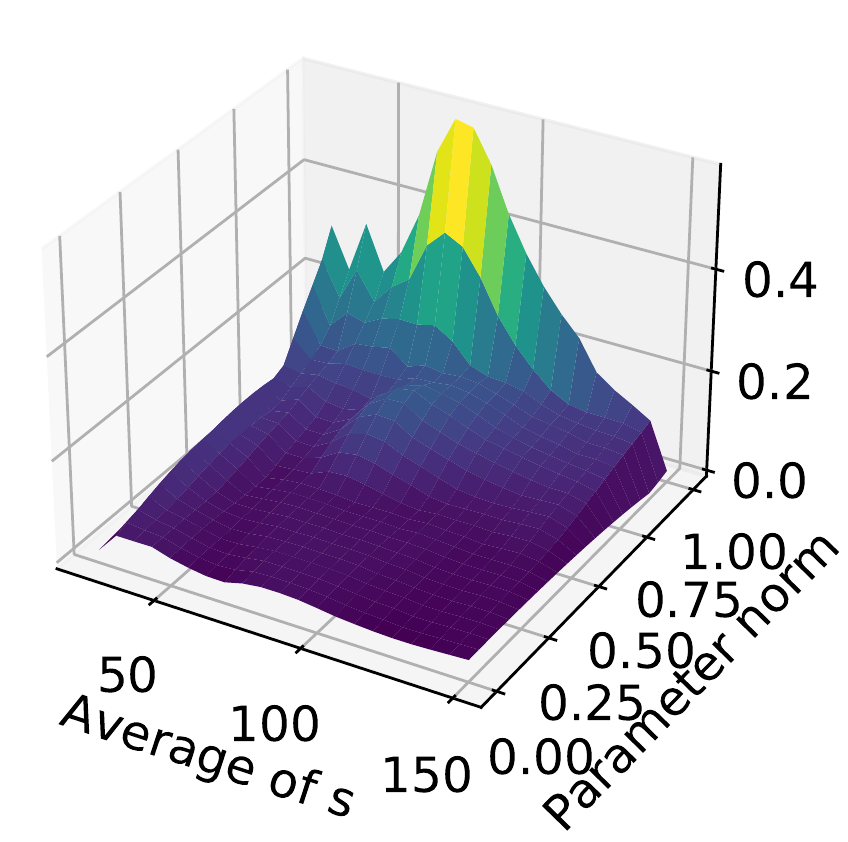}
\hspace{1em}
\includegraphics[width=.3\textwidth]{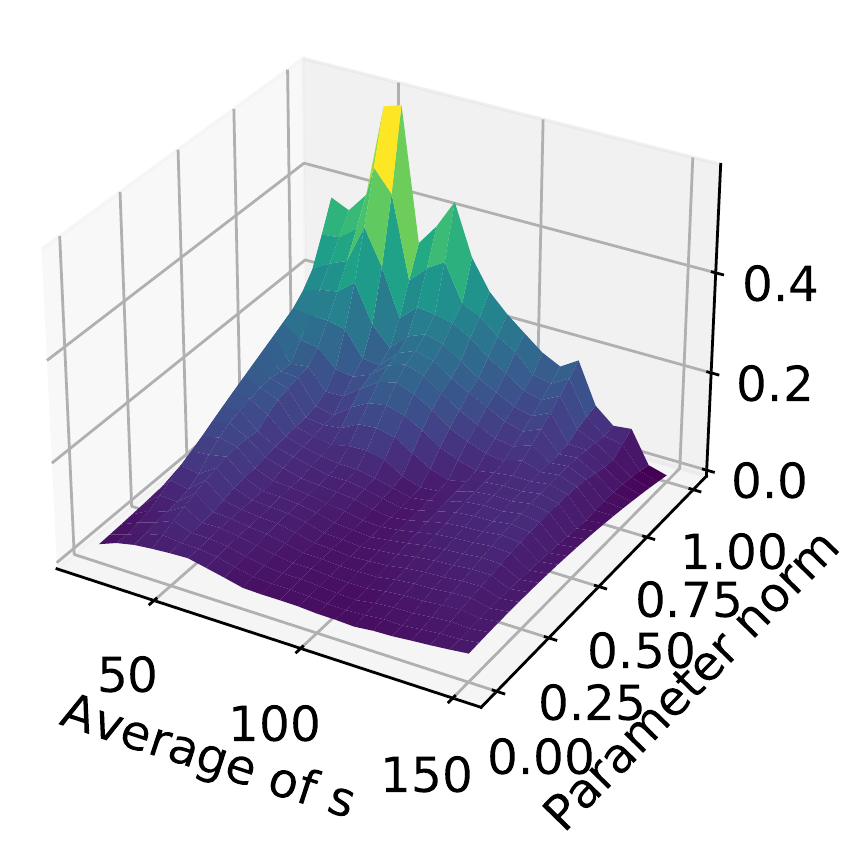}
\hspace{1em}
\includegraphics[width=.3\textwidth]{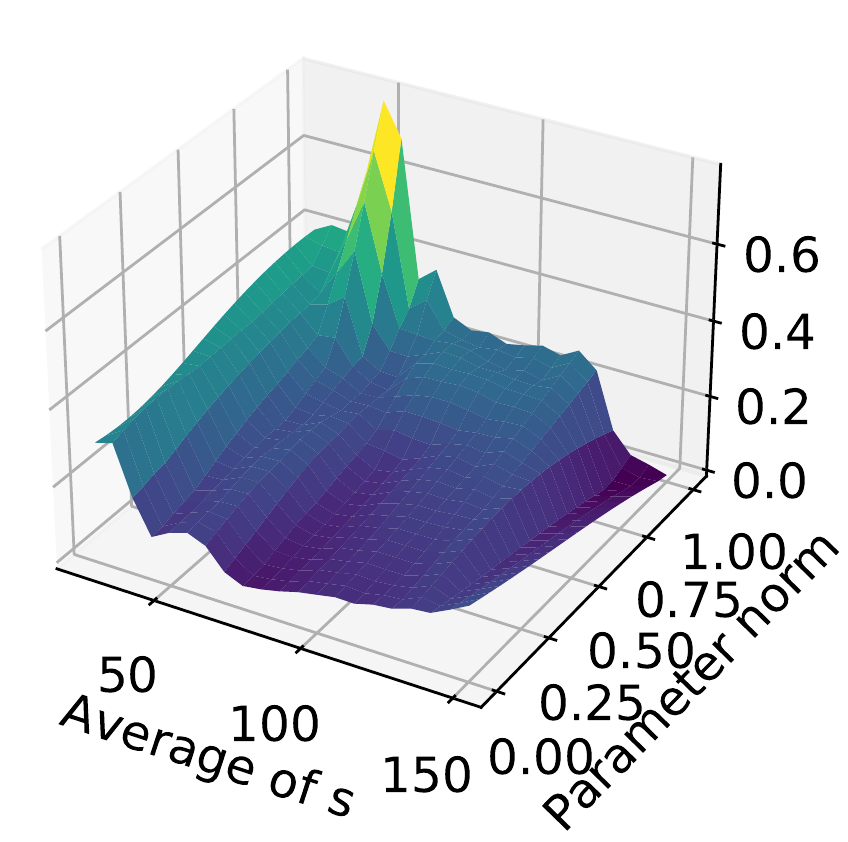}
\caption{Maximal error using geometric payoff depending on the parameter and the average asset price at maximal time to maturity for different dimensions. From left to right $d=3, 5,8$.}
\label{fig:max_error_at_final_time_geometric}
\end{center}
\end{figure}

\begin{figure}[htbp]
\begin{center}
\includegraphics[width=.3\textwidth]{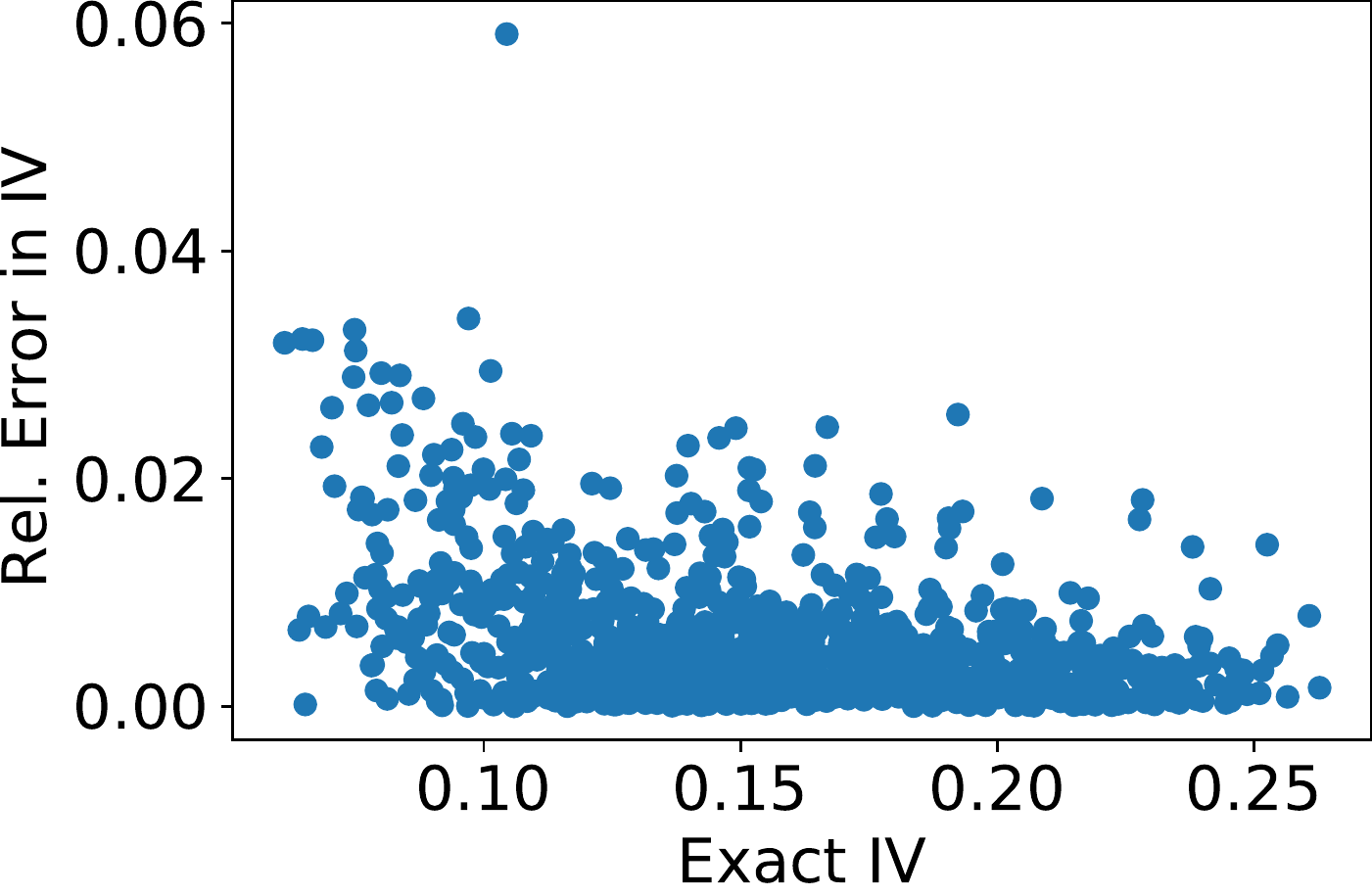}
\includegraphics[width=.3\textwidth]{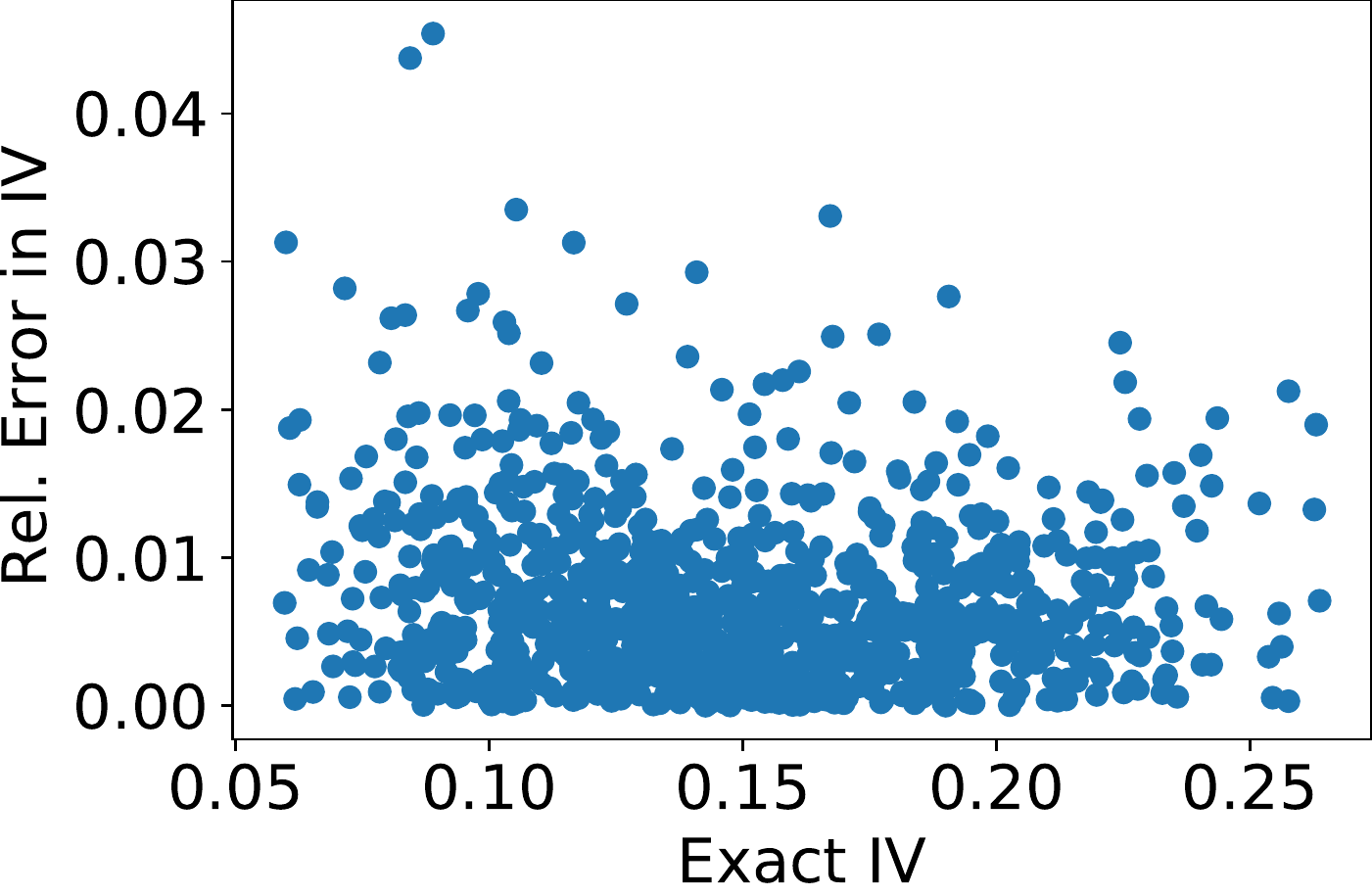}\hspace{1em}
\includegraphics[width=.3\textwidth]{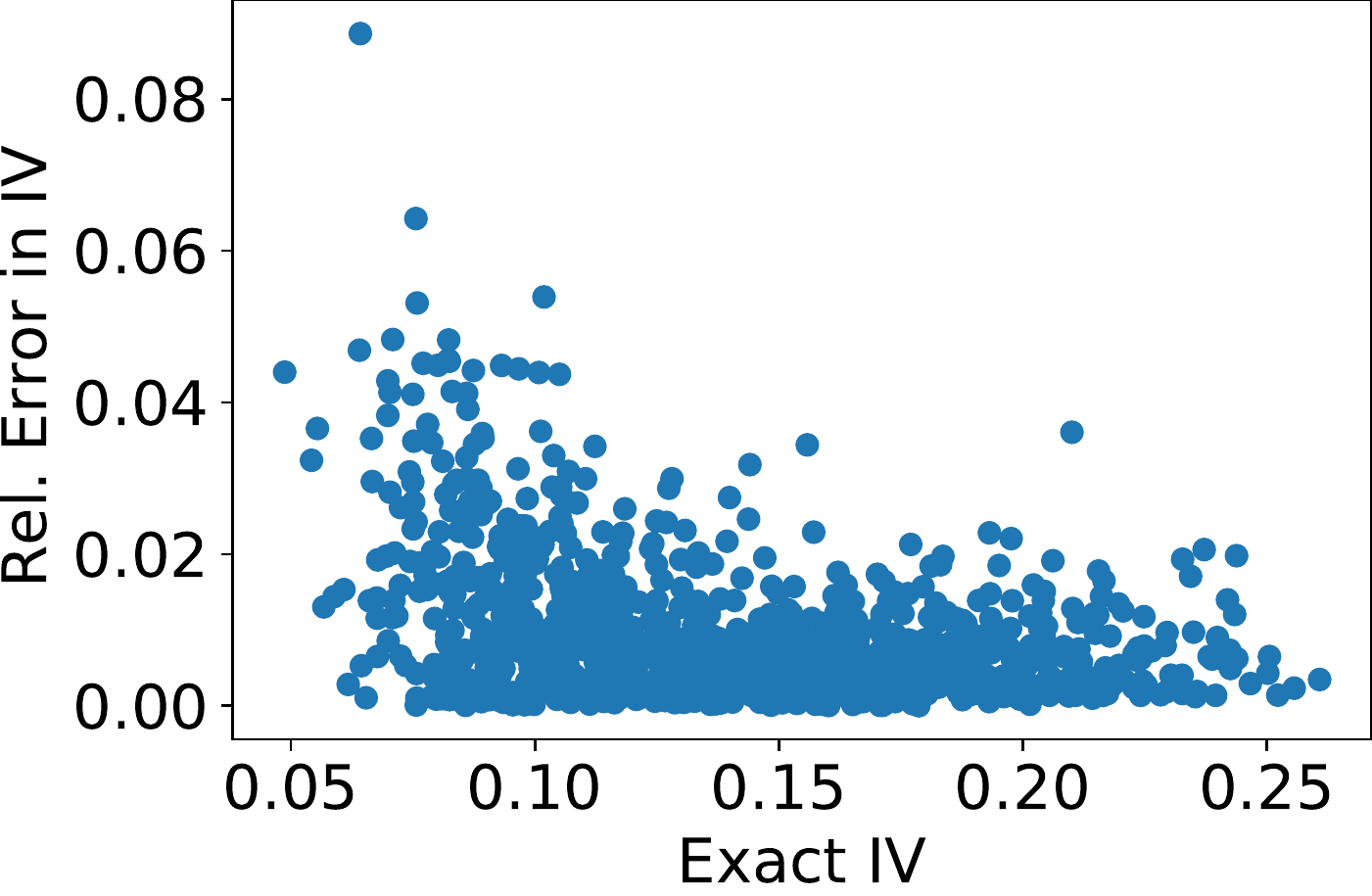}
\caption{Relative implied volatility error for $d=3, 5, 8$ (from left to right)}
\label{fig:scattered_iv_geometric}
\end{center}
\end{figure}
To measure relative errors, we again consider the implied volatility as described in Section~\ref{sec:implied_volatility}. The relative error for $1,000$ random values is shown in Figure~\ref{fig:scattered_iv_geometric}. They show similar error magnitudes compared to the original setting of   basket call options. 

In summary, this shows flexibility of the deep parametric PDE method as we see robust results also in this setting. 

\subsection{Greeks and Sensitivities}
Besides the approximated prices, the Greeks (i.e., derivatives of the price surface) are often of interest for applications, e.g., for hedging. Neural networks provide easy access to these derivatives. Exemplarily,  Figure~\ref{fig:greeks} shows the derivative of the price for two assets with respect to the log-price of the first underlying as well as its accuracy. We see a relative accuracy of below $1\%$  throughout most of the domain. 
\begin{figure}
\begin{center}
\includegraphics[width=.3\textwidth]{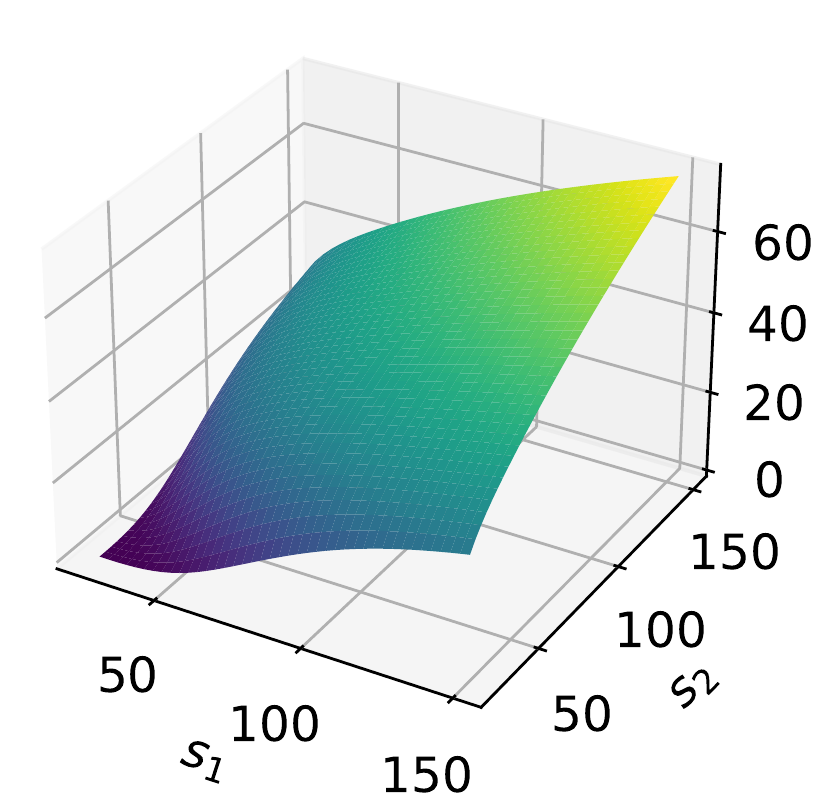}\hspace{1em}
\includegraphics[width=.3\textwidth]{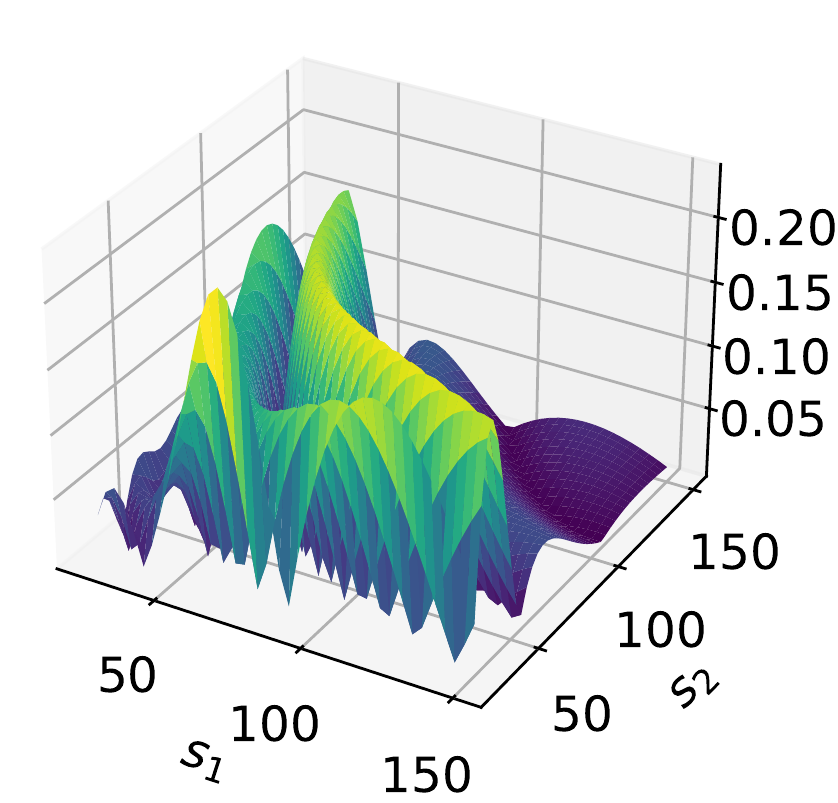}
\caption{Derivative of the price with respect to the log-price of the first underlying (left) and its error (right) for $d=2$.}
\label{fig:greeks}
\end{center}
\end{figure}

Similarly, sensitivities with respect to the parameters can be computed. However, we observed significantly less accuracy. To improve the approximated sensitivities, it might be helpful to add the PDE of the sensitivity to the loss function. This approach is well-defined and will be subject of further research.

\section{Conclusions}\label{sec:conclusion}
We have presented the deep parametric PDE method as an efficient solver for a whole family of parabolic problems. The use of deep neural networks allow us to solve high-dimensional problems accurately and fast. After a single training phase,  the method quickly evaluates the solution at different time, state and parameter values.

Among the applications in science and engineering, we focussed on financial applications and presented results for pricing of basket options in the Black-Scholes model. We have seen a good approximation of the option price over a wide range of parameters.  
The runtimes in the online phase and the accuracy are relatively stable over different dimensions. In the offline phase, the runtimes grow, but do not exhibit a curse of dimensionality. 

Compared to the reference pricer for the basket option with eight assets, the deep parametric PDE method has a speed-up factor of $30$ in the evaluation phase with only a slight reduction of the accuracy. Compared to the deep Galerkin method, accuracy and evaluation time are equivalent with the benefit of having the solution readily available for all parameters. The additional cost in the offline phase pays off, for example in situations where the training can be done in idle times.

These good results motivate future research exploiting a key advantage of the proposed method:
its structure allows for an easy adaptation to a multitude of other and more complex cases. For instance, the PDE can be replaced by  partial integro differential equations or inequalities, thus opening up a large range of applications to different types of options and models, and examples beyond finance. 
The inherent offline-online decomposition shows its full potential in cases where the solution has to be evaluated for many parameters, an example application is the  exposure calculation under uncertain parameters.
\section*{Acknowledgements}
The authors thank Domagoj Demeterfi and Tobias Win\-disch for valuable discussions and Christian P\"otz for   valuable discussions and for providing  the reference option pricer.

This research utilised Queen Mary's Apocrita HPC facility, supported by QMUL Research-IT. doi:10.5281/zenodo.438045
\clearpage
  \bibliographystyle{abbrv} 
  \bibliography{lit}
\end{document}